\documentclass[a4paper,english,10pt]{amsart}

\usepackage{amssymb,amsthm, amsmath,amsfonts}
\usepackage{latexsym}
\usepackage{dsfont}
\usepackage[all]{xy}
\xyoption{arc}

\usepackage{graphicx}
\usepackage{subfigure}
\usepackage{caption}
\usepackage{setspace}
\usepackage[ruled, vlined, linesnumbered]{algorithm2e}
\usepackage{url}
\usepackage{microtype}
\usepackage{hyperref}

\newcommand{\N}{\mathbb{N}}
\newcommand{\Pp}{\mathbb{P}}
\newcommand{\E}{\mathbb{E}}
\newcommand{\F}{\mathcal{F}}

\newcommand{\cals}{\mathcal{S}}
\newcommand{\calp}{\mathcal{P}}
\newcommand{\calk}{\mathcal{K}}
\newcommand{\calm}{\mathcal{M}}
\newcommand{\cala}{\mathcal{A}}

\newcommand{\Rho}{\emph{\textsf{P\hspace{0.1ex}}}}
\newcommand{\Proc}{\rm P}
\newcommand{\n}{\mathbf{n}}

\makeatletter
\DeclareRobustCommand{\rvdots}{%
  \vbox{
    \baselineskip4\p@\lineskiplimit\z@
    \kern-\p@
    \hbox{.}\hbox{.}\hbox{.}
  }}
\DeclareRobustCommand{\srvdots}{%
  \vbox{
    \baselineskip2\p@\lineskiplimit\z@
    \kern-\p@
    \hbox{\scriptsize.}\hbox{\scriptsize.}\hbox{\scriptsize.}
  }}
\DeclareRobustCommand{\srddots}{%
  \vbox{
    \baselineskip2\p@\lineskiplimit\z@
    \kern-\p@
    \hbox{\scriptsize.}\hbox{\scriptsize~.}\hbox{\scriptsize~.}
  }}
\makeatother

\usepackage{pgf,tikz}
\usetikzlibrary{arrows,shapes,calc,positioning,matrix}
\usepackage{tikz-cd}

  \tikzstyle{resident}=[circle,fill=blue,draw=none,text=white]
  \tikzstyle{mutant}=[resident,fill=red]
  \tikzstyle{reproductor}=[circle,very thick, fill=red,draw=green!60!black,text=white]
  \tikzstyle{dead}=[circle,very thick,fill=blue,draw=black,text=white]
  \definecolor{zzzzzz}{rgb}{0.4,0.4,0.4}
  \definecolor{ffttqq}{rgb}{1,0.2,0}
  
\newtheorem{theorem}{\bf Theorem}[section]
\newtheorem{lemma}{\bf Lemma}[section]
\newtheorem{proposition}{\bf Proposition}[section]

\theoremstyle{definition}

\newtheorem{remark}{Remark}

\numberwithin{equation}{section}

\newcommand{\green}{\color{green!60!black}}

\title[Absorption and fixation times]{Absorption and fixation times for \\ evolutionary processes on graphs}

\author[F. Alcalde]{Fernando Alcalde Cuesta}
\address{Universidad de Santiago de Compostela, E-15782 Santiago de Compostela, Spain.
}
\email{fernando.alcaldecuesta@gmail.com}
\thanks{The authors would like to thank Jan Rychtář for suggesting the initial idea of this article. This work has been partially supported by Xunta de Galicia (Spain), \emph{Consolidación e Estruturación} 2023, GRC GI-2136. G.G. acknowledges ANII and CSIC (Uruguay) for financial support. A.L. has been supported by the Spanish Research projects PID2021-122961NB-I00 and PID2022-140556OB-I00 and by the European Regional Development Fund and Diputación General de Aragón (E22-23R)}

\author[G. Guerberoff]{Gustavo Guerberoff}
\address{Instituto de Matem\'atica y Estad\'{\i}stica Rafael Laguardia, 
         Facultad de Ingenier\'{\i}a, Universidad de la Rep\'ublica, 
         J. Herrera y Reissig 565, C.P.11300 Montevideo, Uruguay.}
\email{gguerber@fing.edu.uy}

\author[\'A. Lozano]{\'Alvaro Lozano Rojo}
\address{Departamento de Matemáticas, Instituto Universitario de Matemáticas y Aplicaciones (IUMA), Universidad de Zaragoza, Pedro Cerbuna 12, E-50009 Zaragoza, Spain.}
\email{alozano@unizar.es}

\begin{document}

\begin{abstract}
In this paper, we study the absorption and fixation times for evolutionary processes on graphs, under different updating rules. 
While in Moran process a single neighbour is randomly chosen to be replaced, in proliferation processes other neighbours can be replaced using Bernoulli or binomial draws depending on $0 < p \leq 1$. There is a critical value $p_c$ such that the proliferation is advantageous or disadvantageous in terms of fixation probability depending on whether $p > p_c$ or $p < p_c$. 

We clarify the role of symmetries for computing the fixation time in Moran process. We show that the Maruyama-Kimura symmetry depend on the graph structure induced in each state, implying asymmetry for all graphs except cliques and cycles. There is a fitness value, not necessarily $1$, beyond which the fixation time decreases monotonically. 

We apply Harris' graphical method to prove that the fixation time decreases monotonically depending on $p$. Thus there exists another value $p_t$ for which the proliferation is advantageous or disadvantageous in terms of time. However, at the critical level $p=p_c$, the proliferation is highly advantageous when $r \to +\infty$.
\end{abstract}


\maketitle
  
\section{Introduction}
\label{SIntro}

Proliferation models were introduced in \cite{AlcaldeGuerberoffLozano2022} to more realistically describe the propagation of mutant individuals in a finite population arranged on a graph. Resident individuals evolve like in the structured version \cite{LHN} of the classical Moran process \cite{Moran}, that is, if the individual randomly chosen for reproduction is a resident, then its clone replaces a single neighbour randomly chosen for death. However, if the individual chosen for reproduction is a mutant, we distinguish two types of draw: \emph{Bernoulli}, where all or none of the neighbours become mutants depending on a certain probability $p$; or \emph{binomial}, where the number of neighbours chosen for death follow a binomial distribution with probabilistic parameter $p$. Both processes were 
used to model the spread of prion diseases and the proliferation of stem cells that produce multiple daughter cells, see \cite{AlcaldeGuerberoffLozano2022} and references therein. 

Although it may seem surprising, proliferation is not always to the advantage of mutants. The main result of \cite{AlcaldeGuerberoffLozano2022} states that, whatever the underlying population structure, there always exists a critical value $p_c$ (depending on the graph structure) such that the fixation probability for the proliferation process is equal to the fixation probability for the corresponding Moran process. If $p > p_c$, proliferation is advantageous, but if $p < p_c$ it is disadvantageous. Symbolic computation for a small population shows that this critical value is highly correlated to the mean degree of the graph so that proliferation is more advantageous in cliques than in star graphs. By randomly updating of the proliferation parameter, fluctuating waves in the population of proliferating mutant individuals (similar to the \emph{parasitaemia waves} exhibited in \cite{Gustavo}) can be also created, replacing fixation or extinction by a quasi-stationary state.

In this paper, we are interested in the absorption and fixation times. The \emph{mean absorption time} is the mean number of steps needed to reach one of the two absorbing states of the process, namely fixation or extinction of mutants, starting from any mutant or any set with the same number of mutants. The \emph{mean fixation time} is the mean time required for any mutant (or any set with the same number of mutants) to take over the whole population, assuming that this will happen. 

Similarly to fixation probability, definitions and derivations of these times for Birth-Death processes on (infinite) homogeneous populations can be found in \cite[Chapter 4, Section 7]{KT}.  A subsequent description of finite-state continuous-time Markov chains \cite[Chapter 4, Section 8]{KT} has later been used by Ashcroft et al. \cite{Traulsen2015} and independently by Hathcock and Strogatz \cite{HathcockStrogatz} to study the distribution of mean fixation times in a finite homogeneous population. A full analytic computation has been given by Antal et al. \cite{Antal2006} in the homogeneous case, whereas computations for the Moran process on complete, cycle and star graphs have been derived in \cite{BHR, H}.

Maruyama \cite{MaruyamaTimes1974} and Maruyama and Kimura \cite{MaruyamaKimura} were the first to observe a symmetry between the fixation time of a set of $i$ mutants and the extinction time of the complementary set of $N-i$ residents in a homogeneous population. This symmetry was later rediscovered by Ashcroft et al. \cite{Traulsen2015} and also \cite{Taylor2006} in average. More recently Hathcock and Strogatz \cite{HathcockStrogatz} derived that the fixation time distributions of mutants of fitness $r$ and $1/r$ are identical in this case. 
We will obtain both symmetries from another kind of symmetry in a clique, which we call \emph{G-symmetry}, and we will use it to study the monotonicity of the fixation time of any mutant in the Moran process on any graph.
Although one would expect this time to decrease if the mutant is advantageous and his advantage increases, we will see that this is not strictly true.

Likewise, it could be expected that proliferation would reduce the mean fixation time, at least if the fixation probability is identical for both processes.
Here we will indeed prove that, in the critical regime, proliferation of mutant individuals reduces the mean absorption and fixation times, at least for some large classes of graphs. 
However, we will also show the existence of a new critical probability $p_t$ determining the transition from the regime where proliferation does not reduce the fixation time to the regime where proliferation reduces the fixation time.
 
More precisely, in Section~\ref{SBernoulli},
we explicitly compute absorption and fixation times for complete, cycle and star graphs of order $N$ for Bernoulli proliferation, allowing to compare these times to the corresponding times for Moran process on each graph. The critical case where the proliferation is neither advantageous nor disadvantageous is particularly interesting. In all studied cases, this critical rate (Moran times/proliferation times, at $p=p_c$) is strictly greater than one for any fitness value $r > 0$. The limit of the critical rate coincides for both, absorption and fixation times, as $r \to +\infty$, but depending on the type of graph, this limit is equal to:
\begin{itemize}
\item Complete graph: $(N-1)H_{N-1}$ where $H_{N-1}$ is the $(N-1)$-harmonic number, 

\item Cycle graph: $\frac{\textstyle 2N(N-1)}{\textstyle(N+1)(N-2)}$,

\item Star graph: 
$\frac{\textstyle N}{\textstyle 3N-2} 
\big[ (N-1) \sum_{i=1}^{N-1} H_ i \big].$
\end{itemize}
Analytical expressions and rate limits are completed with the exact computation of times and rates for all graphs of order $6$ included in Supplementary Material S4. Both approaches confirm that mean absorption and fixation proliferation times are lower for star and complete graphs, with cycles being the worst performing graphs.

In Section~\ref{Sbinomial}, we also give complete derivations of mean absorption times for binomial proliferating and non-proliferating individuals on a complete graph. Absorption and fixation times and rates are graphically described for all graphs of order $6$ and also for some additional orders, see Supplementary Material S4-S5.


 Finally, in Sections~\ref{Smonotone}~and~\ref{Sphase}, we demonstrate the existence of a second critical value $p_t$ such that the fixation time rate is greater or less than 1 depending on whether $p$ is greater or less than $p_t$. This imply the existence of three different regimes when we compare Bernoulli or binomial proliferation with the classical Moran process: chance and time disadvantageous, chance or time disadvantageous, and finally chance and time advantageous. 

\section{Proliferation models on evolutionary graphs}
\label{Sdef}

\emph{Proliferation models} on finite graphs $G=(V,E)$ has been introduced in \cite{AlcaldeGuerberoffLozano2022}. All such graphs will be assumed connected and undirected. As for the Moran model \cite{Moran} and its structured version \cite{LHN}, see also \cite{BR,BRS2,Maruyama1970,Maruyama1974,Nowak,SR2,Voorhees1}, we assume: 

\begin{itemize}

\item A finite population is arranged on the vertex
 set $V$ of cardinal $|V|= N$.

\item There exist two types of individuals/particles/alleles/cells: \emph{mutant/resident}, \emph{type-1/type-2}, \emph{advantageous/disadvantageous} or \emph{mutant/wild-type}. 

\item Mutant individuals have a relative fitness $r>0$ with respect to resident individuals. 

\item The population evolves when clones of a mutant or resident individual occupying a vertex $v$ replace individuals, called \emph{neighbours}, occupying vertices $w$ linked to $v$ by an element of the edge set $E$. 
\end{itemize}

\noindent
To define the evolutionary process, we suppose that, at a given time 
$n$, there are $i$ mutants and $N-i$ residents occupying $V$.  Let $S_n \subset V$ be the set of vertices occupied by mutants at this time. A vertex $v$ is randomly 
chosen with probability $p_v =\frac{r}{w_i}$ if $v \in S_n$ and $p_v = \frac{1}{w_i}$ if $v \notin S_n$, where $w_i = ri + N-i$ is the total reproductive weight of all individuals at this state.
As in the Moran process, if the vertex 
$v$ is occupied by a resident then we randomly choose a neighbour $w$ of $v$ and put a resident on the vertex $w$.
If $v$ is occupied by a mutant, under the classical 
model, a neighbour $w$ is also chosen in an equiprobable way to put a mutant. However, in \cite{AlcaldeGuerberoffLozano2022}, we considered two types of clonal proliferation, both associated to a parameter $p \in [0,1]$ which regulates its intensity: 
 
\begin{list}{\itemii}{\leftmargin=14pt}

  \item[i)~] \textbf{Bernoulli proliferation}: with probability $p$ all the 
 neighbours of vertex $v$ (whose number $d(v)$ is the \emph{degree of $v$}) become occupied by mutants, and with probability $1-p$ nothing happens.

  \item[ii)] \textbf{Binomial proliferation}: each neighbour of $v$ is occupied by a mutant with probability $p$ independently. Thus the number of potentially occupied neighbours of $v$ follows a binomial distribution with parameters $d(v)$ and $p$.

\end{list}

The \emph{fixation probability} is the basic quantity for the stochastic analysis of the proliferation studied in \cite{AlcaldeGuerberoffLozano2022}.
For the Moran process on a graph $G = (V,E)$, given a set of vertices $S \subset V$ occupied by mutants, the \emph{fixation  probability}
\[
\Phi^M_S(r)
            = \mathbb{P} \, [ \, \exists n \geq 0 : S_n = V \mid S_0 = S \,]
\]
is the solution of a system of $2^N$ linear equations with boundary conditions 
$\Phi^M_\emptyset(r)=0$ and $\Phi^M_V(r)=1$. Here $S_n$ denotes the set of vertices occupied by mutants at time $n\geq 0$. As $G$ is undirected, such a linear 
system admits a unique solution~\cite{KT:ISM}. 
The \emph{(mean) fixation probability}
is given by 
\[
\Phi^M_1(r) = \frac 1N \sum_{v \in V} \Phi^M_{\{v\}}(r),
\]
being a rational function on the fitness $r$. Similarly, we denote by $\Phi^B_1(r,p)$ and 
$\Phi^b_1(r,p)$ the mean 
fixation probability for Bernoulli and binomial proliferation on 
$G$ with fitness $r$ and proliferation parameter $p$ respectively. In \cite{AlcaldeGuerberoffLozano2022}, we compared the fixation probabilities for proliferating and non-proliferating processes obtaining the following result: 
\begin{theorem}[\cite{AlcaldeGuerberoffLozano2022}]
  \label{thm:critical}
  For any undirected 
  graph $G$, there is a unique pair of 
  critical values $p_c^B = p_c^B(r)$ and $p_c^b = p_c^b(r)$ depending on the 
  fitness $r >0$ such that 
 $$
 \Phi^B_1(r,p_c^B) = \Phi^M_1(r)
\quad \text{and} \quad
\Phi^b_1(r,p_c^b) = \Phi^M_1(r).
$$
\end{theorem}
\noindent
Thus, for every fitness value $r > 0$, there is a transition between two 
different regimes where the proliferation is \emph{advantageous} for mutants if $p > p_c^\beta(r)$, but \emph{disadvantageous} if $p < p_c^\beta(r)$, with $\beta = b,B$. As explained in \cite{AlcaldeGuerberoffLozano2022}, this fact alters the balance between amplifiers, suppressors and isothermal graphs in the Moran 
process~\cite{PLOS,PLOS2018,Allen2020,HindersinTraulsen}.

Here we study other fundamental quantities in the stochastic analysis of the evolution in a finite population: the \emph{mean absorption time} and the \emph{mean fixation time}. 
Denoting again $S_n$ the set of vertices occupied by mutants at time $n\geq 0$, 
for each subset $S$ of $V$, the \emph{absorption time} is the random variable 
$$
\tau^\beta_S = \min \{ \; n\geq 0 \; | \, S_n = \emptyset \text{ or } S_n =V \, \}
$$
conditioned to $S_0 = S$. Its expectation is given as solution of a system of linear equations, although this is now non-homogeneous. With a slight abuse of notation, we denote both the random variable and its expectation by the same symbol. The \emph{mean absorption time} $\tau^\beta_i$ is the average of the expected times $\tau^\beta_S$ associated to the subsets $S$ having $i$ elements.
In the next sections, we will consider the equations in some specific cases, although a precise general description is given in Supplementary Material S1.  
The \emph{fixation time} $T^\beta_S$ is the expectation of the random variable $\tau^\beta_S$, but conditioned to the fixation of mutants (see also Supplementary Material S1). The \emph{mean fixation time} $T^\beta_i$ is again the average of $T^\beta_S$ when $|S|=i$. 
 For convenience, we will usually omit the term mean for the times \mbox{$\tau^\beta_i$ and $T^\beta_i$.}

For comparing the absorption and fixation times $\tau^\beta_1$ and $T^\beta_1$ of a single proliferating mutant (in average) to the absorption and fixation times $\tau^M_1$ and $T^M_1$ of a single non-proliferating mutant in the usual Moran process, we introduce the \emph{absorption} and \emph{fixation rates}

\[
\rho^\beta(r,p) = \frac{\tau^M_1(r)}{\tau^\beta_1(r,p)} \quad \text{and} \quad
\Rho^\beta(r,p) = \frac{T^M_1(r)}{T^\beta_1(r,p)}.
\]
We will be especially interested in the case $p = p_c^\beta(r)$ where proliferation does not imply an advantage or a disadvantage in which respect to the fixation probability, see Sections~\ref{SBernoulli}~and~\ref{Sbinomial}. The indices $M$ and $\beta$ will be suppress when the type of process is clear. At the limit when $r \to +\infty$, we can observe the following properties (which will be proved in Supplementary Material S2): 
\begin{proposition}\label{prop:infinityfitness} 
For any graph $G$, for any state $S \neq \emptyset, V$, and for any $p > 0$, we have:
$$\lim_{r \to +\infty} \Phi_S^M(r) = \lim_{r \to +\infty} \Phi_S^\beta(r,p) = 1,$$ 
where $\beta \in \{b,B\}$. We also have:
$$
\lim_{r \to +\infty} \tau_S^M(r) = \lim_{r \to +\infty} T_S^M(r) \quad \text{and} \quad
\lim_{r \to +\infty} \tau_S^\beta(r,p) = \lim_{r \to +\infty} T_S^\beta(r,p). 
$$
\end{proposition}
\noindent
 These quantities will be later interpreted as the fixation time of the so-called \emph{red process} associated to the Moran process or the proliferation processes, see Section~\ref{Smonotone} and Supplementary Material~S3.

As with the fixation probability, it is natural to be interested in the monotonicity of the mean absorption and fixation times for all processes considered here. Figure~\ref{fig:Moranorder6} shows the difference between unconditional and conditional times for the Moran process. 

\begin{figure}[h!]
  \centering
\subfigure[Mean absorption time $\tau^M_1(r)$]{
  \includegraphics[width=0.42\columnwidth]{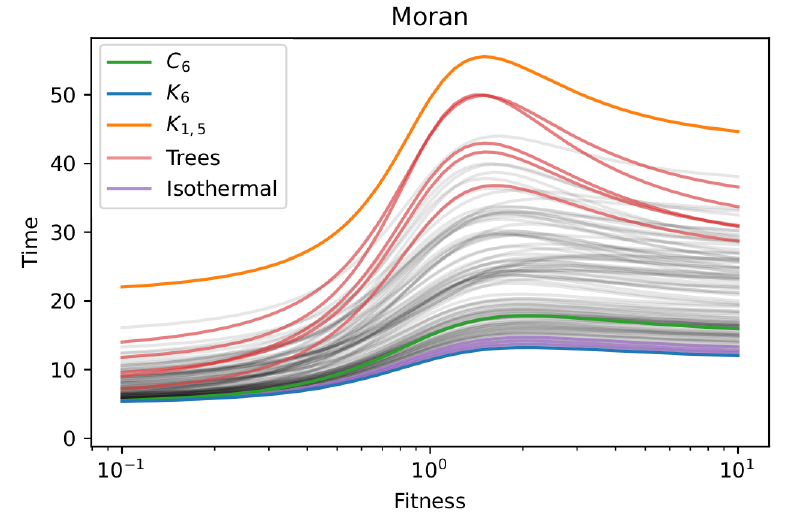}
  \label{fig:moran_unconditional}
}
\subfigure[Mean fixation time $T^M_1(r)$]{
  \includegraphics[width=0.42\columnwidth]{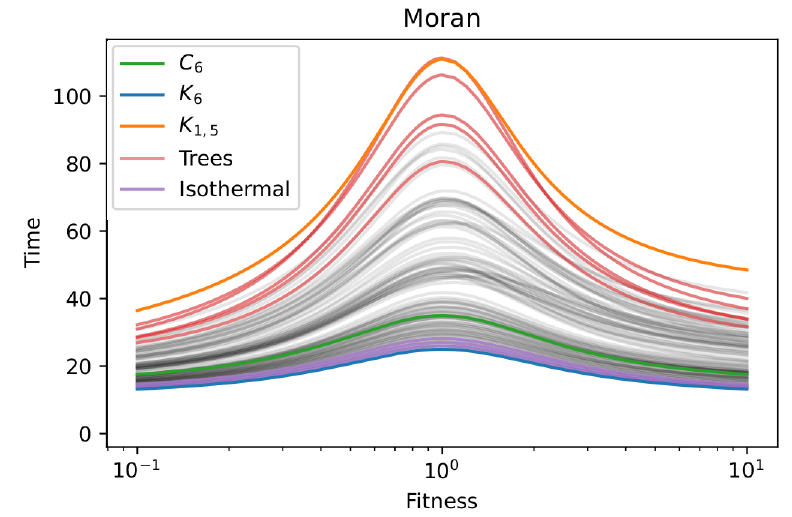}
  \label{fig:moran_conditional}
}
\caption{Mean absorption and fixation times for the Moran process on all graphs of order $6$. 
}
\label{fig:Moranorder6}
\end{figure}
As we will see in Section~\ref{Smonotone}, the symmetry of expected fixation times is only apparent, unless the graph is a clique or a cycle. The following theorem details the general situation:

\begin{theorem} \label{thm:decreaseMoran}
For any graph $G$ and for any vertex $v \in V$, there exist a fitness value $r = r(v) > 0$ depending on $v$ such that the fixation time $T^M_{\{v\}}$ decreases on $[r,+\infty)$. 
\end{theorem} 
For proliferation processes, we will prove the following monotonicity result: 
\begin{theorem} \label{thm:decreasebeta}
For any graph $G$, for any fitness value $r > 0$, and for any vertex $v \in V$, we have that 
$T^\beta_{\{v\}}(r,p') \geq T^\beta_{\{v\}}(r,p)$ 
if $0 < p' \leq p$.  
\end{theorem}
From this result, we will deduce the existence of a second phase transition $p_t$ for the mean fixation time of proliferating mutant individuals: 
\begin{theorem} \label{thm:ptbeta}
For any graph $G$ and for any fitness value $r > 0$, there exists a prolife\-ration parameter $p^\beta_t = p^\beta_t(r)$ such that 
$\Rho^\beta(r,p) \leq 1$ if $p \leq p^\beta_t$ and 
$\Rho^\beta(r,p) \geq 1$ if $p \geq p^\beta_t$.
\end{theorem}

\section{Absorption and fixation times for Bernoulli proliferation}
\label{SBernoulli}

In this section we shall describe the absorption and fixation times for Bernoulli proliferation on cliques, cycles and star graphs. The method we use is quite similar to that used in \cite{BHR,H} to compute the absorption and fixation times for the Moran process on graphs, although Bernoulli proliferation simplifies some computations. 

\subsection{Absorption time for Bernoulli proliferation on complete graphs}

Let $G=K_N$ be the complete graph of order $N$. Since $K_N$ is symmetric, if the starting state $S$ contains $i$ mutants, we can assume $S =\{1,\dots, i\}$ as in  Figure~\ref{fig:initialcomplete}. Then $S$ can only evolve to a state with $i-1$ mutants (equal to $\{1,\dots, i-1\}$ up to symmetry), $S$ or $V$. The non-trivial transition probabilities are given by 
\begin{equation} \label{eq:prob_complete}
\begin{split}
  p_{i,i-1} & =  \frac{\textstyle N-i}{\textstyle w_i} \frac{\textstyle i}{\textstyle N-1}  \\
  p_{i,N} & =  \frac{\textstyle rpi}{\textstyle  w_i}  \\
  p_{i,i}  & =    1 - p_{i,i-1} -  p_{i,N}.
\end{split} 
\end{equation}
For $i=1,2, \ldots, N-1$, the absorption times $\tau_i$ satisfy the equations:
\begin{align} \label{eq:ATcomp}
\begin{split}
\tau_i  & = p_{i,i-1} \tau_{i-1} + p_{i,N} \tau_N + (1 -  p_{i,i-1} -  p_{i,N}) \tau_i + 1 \\
 & = p_{i,i-1} \tau_{i-1} + (1 -  p_{i,i-1} -  p_{i,N}) \tau_i + 1 
\end{split}
\end{align}
as $\tau_N =0$, see Figure~\ref{fig:statecomplete}. Thus we get:
 \begin{align}  \label{eq:eq_abs_complete}
\begin{split}
     \tau_i & = \frac{p_{i,i-1}}{p_{i,i-1} + p_{i,N}} \tau_{i-1} +  \frac{1}{p_{i,i-1} + p_{i,N}} \\
            & = \frac{1 - \frac{i-1}{N-1}}{1+rp - \frac{i-1}{N-1}} \tau_{i-1} + \frac{1}{1+rp - \frac{i-1}{N-1}}\frac{w_i}{i}. 
\end{split}
 \end{align}
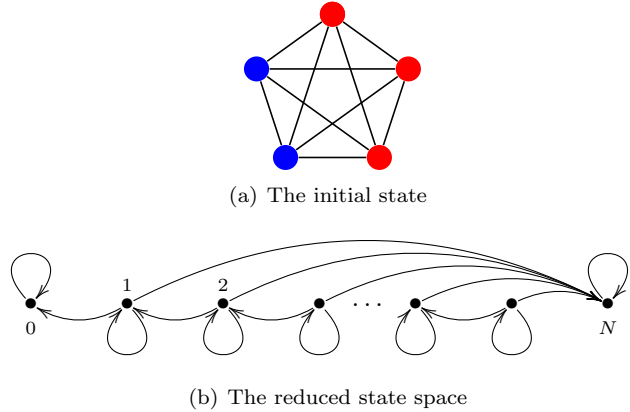
\begin{figure}[b]
  \centering
  \subfigure[The initial state]{
    \label{fig:initialcomplete}
    \begin{tikzpicture}[x=1.0cm,y=1.0cm,scale=0.7]
      \node[mutant] (A0) at (90:1.5)   {};
      \foreach \a in {0,1,2}
        \node[mutant] (A\a) at (90-72*\a:1.5) {};
\foreach \b in {3,4}
        \node[resident] (A\b) at (90-72*\b:1.5) {};

      \foreach  \a in {0,1,2,3}
        \foreach \b in {1,2,3,4}{
          \ifnum\a<\b
            \draw [
            semithick] (A\a) -- (A\b);
          \fi
        }

      \path[use as bounding box] 
        ($(current bounding box.north east)+(1,0)$) rectangle 
          ($(current bounding box.south west)-(1,0)$); 

    \end{tikzpicture}%
  }
  \subfigure[The reduced state space]{
    \label{fig:statecomplete}
    \begin{tikzpicture}[x=1.0cm,y=1.0cm,scale=0.75]
      \clip(-6,-1.2) rectangle (5.3,1.2);
      \node at (-0.5,0) {$\xy
          0;/r.15pc/: 
          (-80,0)*{\bullet}="A0"; 
          (-80,-5)*{\scriptstyle 0};
          (-60,0)*{\bullet}="B0"; 
          (-60,4)*{\scriptstyle 1};
          (-40,0)*{\bullet}="C0"; 
          (-40,4)*{\scriptstyle  2};
          (-20,0)*{\bullet}="D0"; 
          (-10,0)*{\dots};
          (0,0)*{\bullet}="E0"; 
          (20,0)*{\bullet}="F0"; 
          (40,0)*{\bullet}="G0"; 
          (40,-5)*{\scriptstyle N};
          {\ar@/^0.5pc/ "B0"; "A0"}; 
          {\ar@/^0.5pc/ "C0"; "B0"}; 
          {\ar@/^0.5pc/ "D0"; "C0"};  
          {\ar@/^0.5pc/ "F0"; "E0"}; 
          {\ar@/^2pc/ "B0"; "G0"}; 
          {\ar@/^1.6pc/ "C0"; "G0"}; 
          {\ar@/^1.2pc/ "D0"; "G0"}; 
          {\ar@/^0.8pc/ "E0"; "G0"}; 
          {\ar@/^0.4pc/ "F0"; "G0"}; 
          {\ar@(ul,ur) "A0"; "A0"}; 
          {\ar@(dr,dl) "B0"; "B0"}; 
          {\ar@(dr,dl) "C0"; "C0"};  
          {\ar@(dr,dl) "D0"; "D0"};  
          {\ar@(dr,dl) "E0"; "E0"};  
          {\ar@(dr,dl) "F0"; "F0"}; 
          {\ar@(ul,ur) "G0"; "G0"}; 
        \endxy$};
    \end{tikzpicture}
  }
  \caption{Bernoulli proliferation on a complete graph: 
    \subref{fig:initialcomplete}~the initial state $S=\{1,2,3\}$, 
    \subref{fig:statecomplete}~the reduced state space.}
\end{figure}
Reasoning by recurrence, for $i=1,2, \ldots, N-1$, we deduce: 
\begin{equation} \label{eq:abs_complete}
  \tau_i  = \sum_{j=0}^{i-1}  \frac{1}{[ 1+rp- \frac{j}{N-1}]} \frac{w_{j+1}}{j+1}\prod_{k=j+1}^{i-1} \frac{1 - \frac{k}{N-1}}{1+rp - \frac{k}{N-1}}, 
\end{equation}
where the empty product (for $j=i-1$) is equal to $1$. 

This time can be compared with the absorption time for the Moran process
on  the same graph \cite{BHR, H}: 
\begin{equation} \label{eq:abs_Moran1}
  \tau^M_i  = \Phi^M_i \sum_{j=1}^{N-1}  \frac{1}{p^M_{j,j+1}} \sum_{k=j}^{N-1}\prod_{l=j+1}^{k} \gamma_l -\sum_{j=1}^{i-1}  \frac{1}{p^M_{j,j+1}} \sum_{k=j}^{i-1}\prod_{l=j+1}^{k} \gamma_l \\
\end{equation}
where the empty sum (for $i=1$) is equal to $0$. Here $p^M_{j,j+1} = \frac{rj}{w_j}\frac{N-j }{N-1}$ for all $j=1,\dots,N-1$ and $\gamma_l = 1/r$ for all $l=1,\dots,N-1$.The fixation probability for the Moran process is given by
\begin{equation} \label{eq:FP_Moran}
  \Phi_i ^M = \frac{1+\sum_{j=1}^{i-1}\prod_{k=1}^{j} \gamma_k}{1+\sum_{j=1}^{N-1}\prod_{k=1}^{j} \gamma_k} = \frac{\displaystyle 1 - \frac{1}{r^i}}{\displaystyle1 - \frac{1}{r^N}} = r^{N-i} \frac{\displaystyle r^i-1}{\displaystyle r^N-1} .
\end{equation}
So that we have: 
  \begin{align}   \label{eq:abs_Moran}
\begin{split}
    \tau^M_i  &
    =  r^{N-i} \frac{\displaystyle r^i-1}{\displaystyle r^N-1} 
    \Big( \sum_{j=1}^{N-1}  \frac{1}{p^M_{j,j+1}} 
    \frac{\displaystyle 1}{\displaystyle r^{N-j-1}} 
    \frac{\displaystyle r^{N-j} - 1}{\displaystyle r -1} \Big)  
     \\ 
    & - \sum_{j=1}^{i-1}  \frac{1}{p^M_{j,j+1}} \frac{\displaystyle 1}{\displaystyle r^{i-j-1}} \frac{\displaystyle r^{i-j} - 1}{\displaystyle r -1},  \end{split}
\end{align}
where in the last term the empty sum (for $i=1$) remains equal to $0$.
In particular, the \emph{absorption time rate} is
\begin{align} \label{eq:abs_rate}
\begin{split}
\rho(r,p) = \frac{\tau_1^M}{\tau_1~}
& =  \frac{\displaystyle r^{N-1}\frac{\displaystyle r-1}{\displaystyle r^N-1} 
\Big( \sum_{j=1}^{N-1}  \frac{w_j}{rj} \frac{N-1}{N-j} 
\frac{1}{r^{N-j-1}} 
\frac{r^{N-j} - 1}{r -1} \Big)}{\displaystyle \frac{r+N-1}{1+rp}}  \\
 & =  (1+rp)\frac{N-1}{r+N-1} 
\Big( \sum_{j=1}^{N-1}
\frac{rj+N-j}{rj(N-j)} \frac{r^N - r^j}{r^N-1\phantom{o}}  \Big)  \end{split}
\end{align}
and more precisely
\begin{align} 
\label{eq:abs_criticalrate}
\rho(r,p_c) = \frac{r^{N-1} + \cdots + r + 1}{r^{N-2} + \cdots + r +1}\frac{N-1}{r+N-1} 
\Big( \sum_{j=1}^{N-1}
\frac{rj+N-j}{rj(N-j)}\frac{r^N - r^j}{r^N-1\phantom{o}} \Big)
\end{align}
when $p$ is equal to the critical value 
$$p_c = \frac{r^{N-2}}{r^{N-2} + \cdots + r +1}$$
described in \cite{AlcaldeGuerberoffLozano2022}. 
 Although in this case proliferation does not imply an advantage or a disadvantage for the fixation of mutant individuals, we have the following result: 
\begin{proposition} \label{prop:AR_complete}
Let $$H_{N-1} = \sum_{j=1}^{N-1} \frac{1}{N-j}$$ 
be the $(N-1)$-harmonic number. Then the absorption time for a single mutant individual in the Moran process on a complete graph approximates $(N-1)H_{N-1}$ times the absorption time for a single proliferating mutant individual when $p=p_c$ and $r \to +\infty$, that is, 
$\lim_{r \to +\infty} \rho(r,p_c)  = (N-1) H_{N-1}$. \qed
\end{proposition}

\begin{remark} \label{rem:redcomplete}
Similarly, we also deduce that $\lim_{r \to +\infty} \tau_1^M = \lim_{r \to +\infty} T_1^M = (N-1) H_{N-1}$.
\end{remark}

\subsection{Fixation time for Bernoulli proliferation on complete graphs}

As explained in \cite{Antal2006}, \cite{BHR} and \cite{H}, for $i=1,2, \ldots, N-1$, the fixation times $T_i$ satisfy the equations:
\begin{align}   \label{eq:TFcomp}
\begin{split}
    \Phi_i T_i  & = p_{i,i-1} \Phi_{i-1} T_{i-1} + p_{i,N} \Phi_N T_N + (1 -  p_{i,i-1} -  p_{i,N}) \Phi_i T_i  + \Phi_i \\
   & = p_{i,i-1}\Phi_{i-1} T_{i-1} + (1 -  p_{i,i-1} -  p_{i,N}) \Phi_i T_i  + \Phi_i 
\end{split}
 \end{align}
as $T_N =0$. We use the convention $\Phi_{0} T_{0} =0$ as in \cite{H}. As before, we get:

\begin{align}  \label{eq:eq_fix_complete}
\begin{split}
     T_i  & = \frac{p_{i,i-1}}{p_{i,i-1} + p_{i,N}} \frac{\Phi_{i-1}}{\Phi_i} T_{i-1} + 
      \frac{1}{p_{i,i-1} + p_{i,N}} \frac{\Phi_i}{\Phi_i} \\
          &  = \frac{1 - \frac{i-1}{N-1}}{1+rp - \frac{i-1}{N-1}} \frac{\Phi_{i-1}}{\Phi_i} T_{i-1} + 
      \frac{1}{1+rp - \frac{i-1}{N-1}}\frac{w_i}{i}. 
\end{split}
\end{align}
Reasoning by recurrence, for $i=1,2, \ldots, N-1$, we deduce: 
\begin{equation} \label{eq:fix_complete}
  T_i  = \sum_{j=0}^{i-1}  \frac{1}{[1+rp- \frac{j}{N-1}]} \frac{w_{j+1}}{j+1} \frac{\Phi_{j+1}}{\Phi_i} \prod_{k=j+1}^{i-1} \frac{1 - \frac{k}{N-1}}{1+rp - \frac{k}{N-1}}.
\end{equation}
Once again, we can compare to classical situation (also described in \cite{H}) where 
\begin{equation} \label{eq:fix_Moran}
  T^M_i  =  \sum_{j=1}^{N-1}  \frac{\Phi^M_j}{p^M_{j,j+1}} \sum_{k=j}^{N-1}\prod_{l=j+1}^{k} \gamma_l - \frac{1}{\Phi^M_i}\sum_{j=1}^{i-1}  \frac{\Phi^M_j}{p^M_{j,j+1}} \sum_{k=j}^{i-1}\prod_{l=j+1}^{k} \gamma_l.
\end{equation}
The conventions for sums and products are the same as before. Now observe
$$
T_1 = \frac{r+N-1}{1+rp} = \tau_1
$$
as only three states ($1$ and the absorbing states $0$ and $N$) are involved in this computation. As before, the \emph{fixation time rate} is
\begin{align} \label{eq:fix_rate}
\begin{split}
\Rho(r,p) = \frac{T_1^M}{T_1~} & =  (1+rp)\frac{N-1}{r+N-1} 
\Big( \sum_{j=1}^{N-1} \frac{\Phi^M_j}{\Phi^M_1}
\frac{rj+N-j}{rj(N-j)}\frac{r^N - r^j}{r^N-1\phantom{o}}  \Big)    \\ 
& = (1+rp)\frac{N-1}{r+N-1} 
\Big( \sum_{j=1}^{N-1} \Phi^M_j
\frac{rj+N-j}{rj(N-j)} \frac{r^N - r^j\phantom{m}}{r^N- r^{N-1}}  \Big)  \\
& = (1+rp)\frac{N-1}{r+N-1} 
\Big( \sum_{j=1}^{N-1} \frac{r^N(r^j -1)}{r^j(r^N-1)}
\frac{rj+N-j}{rj(N-j)}\frac{r^N - r^j\phantom{m}}{r^N- r^{N-1}}  \Big), \end{split}
\end{align}
and hence the \emph{critical fixation time rate} is given by
\begin{align} \label{eq:fix_criticalrate}
\begin{split}
\Rho(r,p_c)  & =  \frac{r^{N-1} + \cdots + r + 1}{r^{N-2} + \cdots + r +1}\frac{N-1}{r+N-1} 
\big(  \sum_{j=1}^{N-1}
\frac{\Phi^M_j}{\Phi^M_1} \frac{rj+N-j}{rj(N-j)}\frac{r^N - r^j}{r^N-1} \big)  \\ 
 & =  \frac{r^{N-1} + \cdots + r + 1}{r^{N-2} + \cdots + r +1}\frac{N-1}{r+N-1} 
\big( \sum_{j=1}^{N-1} \frac{r^N(r^j -1)}{r^j(r^N-1)}
\frac{rj+N-j}{rj(N-j)}\frac{r^N - r^j\phantom{m}}{r^N- r^{N-1} }  \big). 
\end{split}
\end{align}
Observe that for $r=1$ we get  $\rho(1,p_c) = H_{N-1}$ and $\Rho(1,p_c)= N-1$. Finally, comparing \eqref{eq:fix_criticalrate} with \eqref{eq:abs_criticalrate}, as the rate $\Phi^M_j/\Phi^M_1$ converges to $1$ when $r \to +\infty$ (or applying directly Proposition~\ref{prop:infinityfitness}), we deduce immediately the following result:

\begin{proposition}  \label{prop:FR_complete}
The fixation time for a single mutant individual in the Moran process on a complete graph approximates $(N-1)H_{N-1}$ times the fixation time for a single proliferating mutant individual when $p=p_c$ and $r \to +\infty$, that is, 
$\lim_{r \to +\infty} \Rho(r,p_c)  = (N-1)H_{N-1}$. \qed
\end{proposition}

\noindent
Propositions~\ref{prop:AR_complete}~and~\ref{prop:FR_complete} can be graphically visualised in Supplementary Material S5 for some values of $N$.

\subsection{Absorption time for Bernoulli proliferation on cycles}

For a cycle, we must distinguish when the starting state has one element or more than one element. In the first case, by symmetry, we can assume that the starting state is $S=\{1\}$. It may evolve to states 
with 0 and 3 elements as in Figure~\ref{fig:initialcycle} (or also remains at the same state).
The non-trivial transition probabilities 
are given by 
\begin{align*}
  p_{1,0} &= \frac{1}{w_1} \\
  p_{1,3} &= \frac{rp}{w_1} \\
  p_{1,1} &= 1 - (p_{1,0} + p_{1,3}), 
\end{align*} 
where $w_1 = r + N -1$. As depicted in Figure~\ref{fig:statecycle}, the first equation for the absorption time is given by 
\begin{align} \label{eq:AT1cycle}
\tau_1  
=  p_{1,3} \tau_3+ (1 -  p_{1,0} - p_{1,3}) \tau_1 + 1 
\end{align}
as $\tau_{0}=0$, and then 
\begin{equation} \label{eq:AT3cycle}
\tau_3 = \frac{1+rp}{rp} \tau_1 - \frac{w_1}{rp}.
\end{equation}
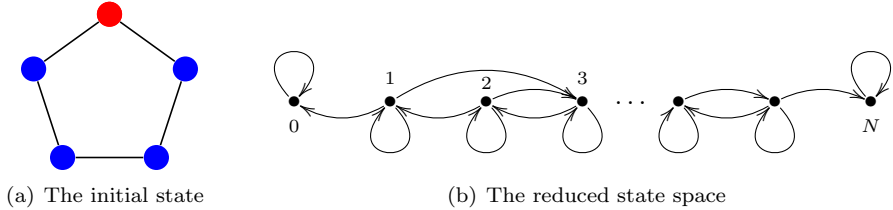
\begin{figure}[b]
  \centering
  \subfigure[The initial state]{
    \label{fig:initialcycle}
    \begin{tikzpicture}[x=1.0cm,y=1.0cm,scale=0.7]
          \foreach \a in {0}
        \node[mutant] (A\a) at (90-72*\a:1.5) {};
\foreach \b in {1,2,3,4}
        \node[resident] (A\b) at (90-72*\b:1.5) {};
      \draw [semithick] (A0)--(A1)--(A2)--(A3)--(A4)--(A0);
      \path[use as bounding box] 
        ($(current bounding box.north east)+(1,0)$) rectangle 
        ($(current bounding box.south west)-(1,0)$); 

\node[mutant] (A0) at (90:1.5)   {};

    \end{tikzpicture}%
  }
  \subfigure[The reduced state space]{
    \label{fig:statecycle}
    \begin{tikzpicture}[x=1.0cm,y=1.0cm,scale=0.75]
      \clip(-6,-1.2) rectangle (5.3,1.2);
      \node at (-0.5,0) {$\xy
        0;/r.15pc/: 
        (-80,0)*{\bullet}="A0"; 
        (-80,-5)*{\scriptstyle 0};
        (-60,0)*{\bullet}="B0"; 
        (-60,5)*{\scriptstyle 1};
        (-40,0)*{\bullet}="C0"; 
        (-40,4)*{\scriptstyle  2};
        (-20,0)*{\bullet}="D0"; 
        (-20,5)*{\scriptstyle  3};
        (-10,0)*{\dots};
        (0,0)*{\bullet}="E0"; 
        (20,0)*{\bullet}="F0"; 
        (40,0)*{\bullet}="G0"; 
        (40,-5)*{\scriptstyle N};
        {\ar@/^0.5pc/ "B0"; "A0"}; 
        {\ar@/^0.5pc/ "C0"; "B0"}; 
        {\ar@/^0.5pc/ "D0"; "C0"};  
        {\ar@/^0.5pc/ "F0"; "E0"}; 
        {\ar@/^1pc/ "B0"; "D0"}; 
        {\ar@/^0.4pc/ "C0"; "D0"}; 
        {\ar@/^0.4pc/ "E0"; "F0"}; 
        {\ar@/^0.4pc/ "F0"; "G0"}; 
        {\ar@(ul,ur) "A0"; "A0"}; 
        {\ar@(dr,dl) "B0"; "B0"}; 
        {\ar@(dr,dl) "C0"; "C0"};  
        {\ar@(dr,dl) "D0"; "D0"};  
        {\ar@(dr,dl) "E0"; "E0"};  
        {\ar@(dr,dl) "F0"; "F0"}; 
        {\ar@(ul,ur) "G0"; "G0"}; 
      \endxy$};
    \end{tikzpicture}%
  }
  \caption{Bernoulli proliferation on a cycle:
    \subref{fig:initialcycle}~the initial state $S=\{1\}$, 
    \subref{fig:statecycle}~the reduced state space.}
\label{fig:cycle}
\end{figure}
In general, as in the usual Moran model, all individuals of the same type are contiguous. Thus, for $i=2,3, \ldots , N-1$, the transition probabilities are given by:
  \begin{align}  \label{prob_cycle_2}
\begin{split}
    p_{i,i-1} &= \frac{\textstyle 1}{\textstyle w_i} \medskip  \\
    p_{i,i+1} &= \frac{\textstyle 2rp}{\textstyle w_i} \medskip  \\
    p_{i,i} &= 1 - (p_{i,i-1} + p_{i,i+1}), 
\end{split}
  \end{align}
where $w_i = ri + N - i$.
The corresponding absorption times are given by the equations:
\begin{align} \label{eq:AT_eq_cycle}
\tau_i  & = p_{i,i-1} \tau_{i-1} + p_{i,i+1} \tau_{i+1}+ (1 -  p_{i,i-1} -  p_{i,i+1}) \tau_i + 1.
\end{align}
As usually (see \cite{H} and \cite{AGSL}), we can solve this system by replacing
$\delta_i = \tau_i - \tau_{i-1}$
to obtain
\begin{align} \label{eq:ATsubs:cycle1}
\delta_{i+1} = \frac{1}{2rp}\delta_i -  \frac{w_i}{2rp}
\end{align}
for $i \geq 2$. More generally, we have 
\begin{align} \label{eq:ATsubs:cycle}
\delta_{i+1} = \big(\frac{1}{2rp}\big)^{i-1}\delta_2 - \sum_{j=2}^i
\big(\frac{1}{2rp}\big)^{i-j+1}w_j
\end{align}
for $i=2,3, \ldots , N-1$. Then
\begin{align} \label{eq:ATcycle1}
0 = \tau_N - \tau_0  & = \sum_{i=1}^N \delta_i  = \Big( \sum_{i=0}^{N-2} \big(\frac{1}{2rp}\big)^i \Big) \delta_2  -  \sum_{i=2}^{N-1} \sum_{j=2}^i \big(\frac{1}{2rp}\big)^{i-j+1}w_j + \delta_1.
\end{align}
Combining
$$
\tau_2 = \frac{1}{1+2rp} \tau_1 + \frac{2rp}{1+2rp}  \tau_3 + \frac{w_2}{1+2rp} 
$$

\noindent
with \eqref{eq:AT3cycle}, we obtain 
$$
\tau_2 = \Big( 1 + \frac{2}{1+2rp} \Big) \tau_1 - 
 \frac{2w_1}{1+2rp} + \frac{w_2}{1+2rp} 
$$
and hence 
$$
\delta_ 2 = \tau_2 - \tau_1 =\frac{2}{1+2rp} \tau_1 + \frac{w_2 - 2w_1}{1+2rp} = \frac{2}{1+2rp} \tau_1 -  \frac{N}{1+2rp}.
$$
Therefore \eqref{eq:ATcycle1} becomes:
\begin{align} \label{eq:ATcycle}
\begin{split}
 \Big(1 + \frac{2}{1+2rp}  \sum_{i=0}^{N-2} \big(\frac{1}{2rp}\big)^i \Big) \tau_1 & =   \sum_{i=2}^{N-1} \sum_{j=2}^i \big(\frac{1}{2rp}\big)^{i-j+1}w_j \\
& + \frac{N}{1+2rp} \sum_{i=0}^{N-2} \big(\frac{1}{2rp}\big)^i  
\end{split}
\end{align}
and finally we get: 
\begin{equation} \label{eq:ATCycle}
 \tau_1 =  \frac{\displaystyle \sum_{i=2}^{N-1} \sum_{j=2}^i \big(\frac{1}{2rp}\big)^{i-j+1}w_j + \frac{N}{1+2rp} \sum_{i=0}^{N-2} \big(\frac{1}{2rp}\big)^i}{\displaystyle 1 + \frac{2}{1+2rp}  \sum_{i=0}^{N-2} \big(\frac{1}{2rp}\big)^i}.
\end{equation}
Like for complete graphs, this value can be compared with the value of the absorption time for one single mutant in the Moran process on the cycle: 
\begin{align} \label{eq:ATMorancycle}
 \tau_1^M  =  r^{N-1} \frac{\displaystyle r -1}{\displaystyle r^N - 1}  
\Big( \sum_{i=1}^{N-1} \sum_{j=1}^i \frac{w_j}{r^{i-j+1}} \Big)   = 
\frac{\displaystyle r^{N-1}}{r^{N-1} + \cdots + r + 1} \Big( \sum_{i=1}^{N-1} \sum_{j=1}^i \frac{w_j}{r^{i-j+1}} \Big). 
\end{align}
As is shown in \cite{AlcaldeGuerberoffLozano2022}, the critical value $p_c$ is the only solution of the equation
\begin{align} 
\frac{2}{1+2rp} \sum_{i=0}^{N-2} \big(\frac{1}{2rp}\big)^i = \sum_{i=1}^{N-1} \big( \frac{1}{r} \big)^i.
\end{align}
Then the absorption time rate at $p=p_c$ becomes, after some simple calculations:
\begin{align} \label{eq:CriticalAb}
\rho(r,p_c) = \frac{\displaystyle \sum_{i=1}^{N-1} \sum_{j=1}^i \big(\frac{1}{r}\big)^{i-j+1}\big(rj + N-j \big)}{\displaystyle \sum_{i=2}^{N-1} \sum_{j=2}^i \big(\frac{1}{2rp_c}\big)^{i-j+1}\big(rj + N-j \big) + \frac{N}{2} \sum_{i=1}^{N-1} \big(\frac{1}{r}\big)^i}.
\end{align}
Observe that, for $r=1$ and $p_c=\frac{1}{2}$, we get:
\[
\rho(1,p_c) = \frac{N \sum_{i=1}^{N-1} i}{N \sum_{i=2}^{N-1}(i-1) + \frac{N}{2} (N-1)} = \frac{N}{N-1}.
\]
When $r \rightarrow + \infty$, 
the situation is now somewhat different from that described in Proposition~\ref{prop:AR_complete}: 

\begin{proposition} \label{prop:AR_cycle}
The absorption time rate for a single mutant individual on a cycle graph satisfies the condition:  
$$
\lim_{r \to +\infty} \rho(r,p_c)  = \frac{2N(N-1)}{(N+1)(N-2)}.
$$
\end{proposition}

\begin{proof}   
 From equations \eqref{eq:ATCycle} and \eqref{eq:ATMorancycle}, we deduce $\tau_1^M/\tau^1$ is the product of the following three factors:

\begin{equation} \label{eq:factor1}
  1 + \frac{2}{1+2rp}  \frac{ (2rp)^{N-2} + \cdots + 2rp + 1}{ (2rp)^{N-2}},
\end{equation}
\begin{equation} \label{eq:factor2}
\frac{r^{N-1}}{ r^{N-1} + \cdots + r + 1},
\end{equation}
\begin{align} \label{eq:factor3}
\begin{split}
\phantom{O} & \phantom{=}
  \frac{\displaystyle 
   \sum_{i=1}^{N-1} \sum_{j=1}^i \frac{w_j}{r^{i-j+1}} }{\displaystyle \sum_{i=2}^{N-1} \sum_{j=2}^i \big(\frac{1}{2rp}\big)^{i-j+1}w_j + \frac{N}{1+2rp} \sum_{i=0}^{N-2} \big(\frac{1}{2rp}\big)^i} \\
    & = \frac{\displaystyle  
 \sum_{i=2}^{N-1} \sum_{j=2}^i \frac{rj+N-j}{r^{i-j+1}} + 
\sum_{i=1}^{N-1}  \frac{r+N-1}{r^i}
}{\displaystyle 
\sum_{i=2}^{N-1} \sum_{j=2}^i \frac{rj+N-j}{(2rp)^{i-j+1}} + \frac{N}{1+2rp} \frac{(2rp)^{N-2} + \cdots + 2rp + 1}{(2rp)^{N-2}}}  \\
&  = \frac{\displaystyle  D(r) + \sum_{i=1}^{N-1}  \frac{r+N-1}{r^i}}{\displaystyle D(2rp) + \frac{N}{1+2rp} \frac{(2rp)^{N-2} + \cdots + 2p + 1}{(2rp)^{N-2}}}. 
\end{split}
\end{align}
The two first factors \eqref{eq:factor1} and \eqref{eq:factor2} converge to $1$ as $r \to +\infty$. For the third factor \eqref{eq:factor3}, we have: $$
\lim_{r \to +\infty} D(r)= \frac{(N+1)(N-2)}{2}
\quad \text{and} \quad 
\lim_{r \to +\infty} D(2rp) = \frac{1}{2p} \frac{(N+1)(N-2)}{2}, 
$$
so this factor converges to 
$$
\frac{2pN(N-1)}{(N+1)(N-2)}.$$
Therefore we have:
$$
\lim_{r \to +\infty} \rho(r,p)  = \frac{2pN(N-1)}{(N+1)(N-2)}.
$$
Finally, as $p_c$ converges to $1$ as $r \to +\infty$, we deduce:
$$
 \hspace{1.35in} \lim_{r \to +\infty} \rho(r,p_c)  = \frac{2N(N-1)}{(N+1)(N-2)}. \hspace{1.87in} \qedhere
$$
\end{proof} 

\begin{remark} \label{rem:redcycle}
From \eqref{eq:ATMorancycle}, we also deduce that $\lim_{r \to +\infty} \tau_1^M = \lim_{r \to +\infty} T_1^M = \frac{N(N-1)}{2}$.
\end{remark}

\subsection{Fixation time for Bernoulli proliferation on cycles}
As before, first we write the equation for the fixation time of one single mutant: 
  \begin{align}  \label{eq:TF_eq_cycle1}
    \Phi_1 T_1  & = p_{1,0} \Phi_{0} T_{0} + p_{1,3} \Phi_3 T_3 + (1 -  p_{1,0} -  p_{1,3}) \Phi_1 T_1 + \Phi_1 \\
      &= p_{1,3} \Phi_3 T_3 + (1 -  p_{1,0} -  p_{1,3}) \Phi_1 T_1 + \Phi_1 \nonumber
  \end{align}
as $\Phi_0T_0 = 0$. We deduce: 
\begin{align} \label{eq:TF3cycle}
T_3  =   \frac{\Phi_1}{\Phi_3} \big( \frac{1+rp}{rp} T_1 - \frac{w_1}{rp}\big). 
\end{align}
In general, to solve the equation 
\begin{align} \label{eq:TFcycle}
\Phi_i T_i  & = p_{i,i-1} \Phi_{i-1} T_{i-1} + p_{i,i+1} \Phi_{i+1} T_{i+1} + (1 -  p_{i,i-1} -  p_{i,i+1}) \Phi_i T_i  + \Phi_i
\end{align}
for $2 \leq i \leq N-1$, we proceed in two steps: 
\smallskip 

\noindent
(i) First, we replace $z_i = \Phi_iT_i$ to obtain: 
\begin{align} \label{eq:TFsimplecycle1}
z_i & = p_{i,i-1} z_{i-1} + p_{i,i+1} z_{i+1} + (1 -  p_{i,i-1} -  p_{i,i+1}) z_i  + \Phi_i.
\end{align}

\noindent
(ii) Second, we replace $\delta_i = z_i - z_{i-1}$ to get:
\begin{equation}
  \label{eq:TFsimplecycle2}
  \delta_{i+1} = \frac{1}{2rp}\delta_i - \frac{\Phi_i}{p_{i,i+1}} 
  = \frac{1}{2rp}\delta_i - \frac{w_i\Phi_i}{2rp}.
\end{equation}
As before, we deduce: 
\begin{align} \label{eq:TFcyclesum}
0 = z_N - z_0  & = \sum_{i=1}^N \delta_i  = \Big( \sum_{i=0}^{N-2} \big(\frac{1}{2rp}\big)^i \Big) \delta_2  -  \sum_{i=2}^{N-1} \sum_{j=2}^i \big(\frac{1}{2rp}\big)^{i-j+1}w_j\Phi_j+ \delta_1.
\end{align} 
Combining again 
$$
z_2 = \frac{1}{1+2rp} z_1 + \frac{2rp}{1+2rp} z_3 + \frac{w_2\Phi_2}{1+2rp} 
$$
with \eqref{eq:TF3cycle}, we obtain 
$$
z_2 = \big(1+ \frac{2}{1+2rp} \big) z_1 - 
 \frac{2w_1\Phi_1}{1+2rp} + \frac{w_2\Phi_2}{1+2rp} 
$$
and hence 
$$
\delta_ 2 = z_2 - z_1 =\frac{2}{1+2rp} z_1 + \frac{w_2\Phi_2- 2w_1\Phi_1}{1+2rp}.
$$
Therefore \eqref{eq:TFcyclesum} becomes:
\begin{equation} \label{eq:ATfinalcycle}
 \Big(1 + \frac{2}{1+2rp}  \sum_{i=0}^{N-2} \big(\frac{1}{2rp}\big)^i \Big) z_1 
=  \sum_{i=2}^{N-1} \sum_{j=2}^i \big(\frac{1}{2rp}\big)^{i-j+1}w_j\Phi_j  
+ \frac{2w_1\Phi_1 - w_2\Phi_2 }{1+2rp} \sum_{i=0}^{N-2} \big(\frac{1}{2rp}\big)^i  
\end{equation}
and finally we get: 
\begin{equation} \label{eq:TFCycle}
 T_1 =  \frac{\displaystyle \sum_{i=2}^{N-1} \sum_{j=2}^i \big(\frac{1}{2rp}\big)^{i-j+1}w_j\frac{\Phi_j}{\Phi_1} + \frac{ 2w_1 - w_2\frac{\Phi_2}{\Phi_1}}{1+2rp} \sum_{i=0}^{N-2} \big(\frac{1}{2rp}\big)^i}{\displaystyle  1 + \frac{2}{1+2rp}  \sum_{i=0}^{N-2} \big(\frac{1}{2rp}\big)^i }.
\end{equation}
Using a similar argument, the identity
\begin{equation*}
0   = \sum_{i=1}^N \delta_i  = \Big( \sum_{i=0}^{N-1} \big(\frac{1}{r}\big)^i \Big) \delta_1  -  \sum_{i=1}^{N-1} \sum_{j=1}^i \big(\frac{1}{r}\big)^{i-j+1}w_j\Phi^M_j, 
\end{equation*} 
implies that the fixation time for one single non-proliferating mutant is given by
\begin{equation} \label{eq:TFMorancycle}
 T_1^M    =  \frac{\delta _1}{\Phi^M_1} = 
 \sum_{i=1}^{N-1} \sum_{j=1}^i \frac{w_j}{r^{i-j+1}} \Phi^M_j.
\end{equation}
For $r=1$ and $p_c = \frac{1}{2}$, the critical fixation rate results (after using $w_i=N$ and $\Phi^M_i=\Phi_i = \frac{i}{N}$, for all $1 \leq i \leq N-1$, and some simple algebra):
\[
\Rho(1,p_c) = \frac{N(N+1)}{(N-2)(N+3)}.
\]
The same argument used before allows us to state the following result: 
\begin{proposition} \label{prop:FR_cycle}
The fixation time rate for a single mutant individual on a cycle graph satisfies the condition:
$$
\lim_{r \to +\infty} \rho(r,p_c)  = \frac{2N(N-1)}{(N+1)(N-2)}. \qedhere
$$
\end{proposition}
\noindent
Propositions~\ref{prop:AR_cycle}~and~\ref{prop:FR_cycle} can be also graphically visualised in Supplementary Material S5. 

\subsection{Absorption time for Bernoulli proliferation on stars}

Using notations from~\cite{AGSL}, the starting state will be denoted by $(0,1)$ or $(1,0)$ depending on whether the mutant is placed in the central vertex (as in Figure~\ref{fig:initialstar}) or in some peripheral 
vertex. The transition probabilities depend on the pair $(i,\varepsilon)$ where $i$ is the number of peripheral vertices occupied by mutants and $\varepsilon = 0,1$ depending on the type of individual that occupies the centre, see Figure~\ref{fig:statestar}. We write $p^\pm_{i,\varepsilon}$  if the number of mutants increases or decreases. The total reproductive weight depend only on the total number of mutants, namely $j = i+\varepsilon$, and is denoted $w_j$. 
Thus, if we denote 
$\hat p^\pm_{i,\varepsilon}$ 
the transition probabilities for the embedded jump process (see Supplementary Material S2), then
\begin{align} \label{eq:star_11}
\begin{split}
  \tau_{0,1} &= \hat p_{0,1}^- \tau_{0,0} 
                  + \hat p_{0,1}^+ \tau_{N-1,1} + 
			\frac{1}{p_{0,1}^- +p_{0,1}^+}    \\ 
             & = \frac{N-1}{N-1+rp}\tau_{0,0} 
                  + \frac{rp}{N-1+rp}\tau_{N-1,1} + \frac{w_1}{N-1+rp} \phantom{xxxX} \\
             & = \frac{w_1}{N-1+rp}, 
\end{split}
\end{align}
\begin{align} \label{eq:star_12}
\begin{split}
\tau_{1,0} &= \hat p_{1,0}^- \tau_{0,0} 
                  + \hat p_{1,0}^+ \tau_{1,1} +  
	\frac{1}{p_{1,0}^- +p_{1,0}^+}  \\ 
             &= \frac{1}{(N-1)rp+1} \tau_{0,0} 
                  + \frac{(N-1)rp}{(N-1)rp+1}\tau_{1,1} + 
\frac{(N-1)w_1}{(N-1)rp+1}  \\
  &= \frac{(N-1)rp}{(N-1)rp+1}\tau_{1,1} + 
\frac{(N-1)w_1}{(N-1)rp+1}, 
\end{split}
\end{align}
\begin{align} \label{eq:star_13}
\begin{split}
\tau_{1,1} &= \hat p_{1,1}^- \tau_{1,0} + \hat p_{1,1}^+ \tau_{N-1,1}+ \frac{1}{p_{1,1}^- +p_{1,1}^+}  \\    
 &= \frac{N-2}{N-2+ rp} \tau_{1,0} 
                 + \frac{rp}{N-2+ rp}\tau_{N-1,1} + \frac{w_2}{N-2+rp}   \phantom{xxxX} \\
       &= \frac{N-2}{N-2+ rp} \tau_{1,0} + \frac{w_2}{N-2+rp}. 
\end{split}
\end{align}
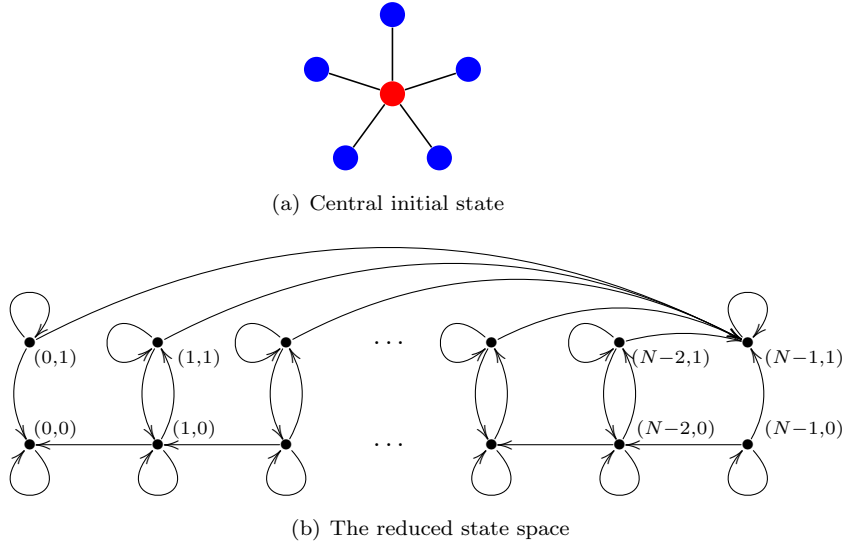
\begin{figure}[t]
\centering
\subfigure[Central initial state]{
\label{fig:initialstar}
 \begin{tikzpicture}[x=1.0cm,y=1.0cm,scale=0.7]
\clip(-4,-1.6) rectangle (4,1.8);
        \node[resident] (A)  at (90:1.5)  {};
        \foreach \a/\b in {1/B, 2/C, 3/D, 4/E}{
          \node[resident] (\b) at (90+72*\a:1.5) {};
        }
        \node[mutant] (O) at (0,0) {};
        \foreach \to in {A, B, C, D, E}
          \draw [semithick] (O) -- (\to);
      \end{tikzpicture}}
\subfigure[The reduced state space]{
\label{fig:statestar}
\begin{tikzpicture}[x=1.0cm,y=1.0cm,scale=0.78]
\clip(-8.5,-2.1) rectangle (9,2.3);
\node at (0,0) {
    $
     \xy
      0;/r.16pc/: 
      (-70,0)*{\bullet}="A0"; 
      (-65,3)*{\scriptstyle (0,0)};
      (-45,0)*{\bullet}="B0"; 
      (-38,3)*{\scriptstyle (1,0)};
      (-20,0)*{\bullet}="C0"; 
      (0,0)*{\dots}="D0"; 
      (20,0)*{\bullet}="E0"; 
      (45,0)*{\bullet}="F0"; 
      (56,3)*{\scriptstyle (N-2,0)};
      (70,0)*{\bullet}="G0"; 
      (81,3)*{\scriptstyle (N-1,0)};
      (-70,20)*{\bullet}="A1"; 
      (-65,17)*{\scriptstyle (0,1)};
      (-45,20)*{\bullet}="B1"; 
      (-37,17)*{\scriptstyle (1,1)};
      (-20,20)*{\bullet}="C1"; 
      (0,20)*{\dots}="D1"; 
      (20,20)*{\bullet}="E1"; 
      (45,20)*{\bullet}="F1"; 
      (55,17)*{\scriptstyle (N-2,1)};
      (70,20)*{\bullet}="G1"; 
      (81,17)*{\scriptstyle (N-1,1)};
      {\ar@{->} "B0"; "A0"}; 
      {\ar@{->} "C0"; "B0"}; 
      {\ar@{->} "F0"; "E0"}; 
      {\ar@{->} "G0"; "F0"}; 
      {\ar@/^3pc/"A1"; "G1"}; 
      {\ar@/^2.5pc/ "B1"; "G1"}; 
      {\ar@/^2pc/"C1"; "G1"}; 
      {\ar@/^1.1pc/"E1"; "G1"}; 
      {\ar@/^0.3pc/"F1"; "G1"}; 
      {\ar@/_0.5pc/ "A1"; "A0"}; 
      {\ar@/_0.5pc/ "G0"; "G1"}; 
      {\ar@/_0.5pc/ "B1"; "B0"}; 
      {\ar@/_0.5pc/ "B0"; "B1"}; 
      {\ar@/_0.5pc/ "C1"; "C0"}; 
      {\ar@/_0.5pc/ "C0"; "C1"}; 
     {\ar@/_0.5pc/ "E1"; "E0"}; 
      {\ar@/_0.5pc/ "E0"; "E1"}; 
      {\ar@/_0.5pc/ "F1"; "F0"}; 
      {\ar@/_0.5pc/ "F0"; "F1"}; 
      {\ar@(ul,ur) "A1"; "A1"}; 
      {\ar@(lu,ld) "B1"; "B1"}; 
      {\ar@(lu,ld) "C1"; "C1"}; 
      {\ar@(lu,ld)  "E1"; "E1"}; 
      {\ar@(lu,ld)  "F1"; "F1"}; 
      {\ar@(ul,ur) "G1"; "G1"}; 
      {\ar@(dr,dl) "A0"; "A0"}; 
      {\ar@(dr,dl) "B0"; "B0"}; 
      {\ar@(dr,dl) "C0"; "C0"}; 
     {\ar@(dr,dl) "E0"; "E0"}; 
      {\ar@(dr,dl) "F0"; "F0"}; 
      {\ar@(dr,dl) "G0"; "G0"}; 
    \endxy
  $};
    \end{tikzpicture}}
  \caption{Bernoulli proliferation on a star:
    \subref{fig:initialstar}~central initial state, 
    \subref{fig:statestar}~the reduced state space.}
\end{figure}
From \eqref{eq:star_12} and \eqref{eq:star_13}, we deduce:
\[
  \tau_{1,0} = \frac{(N-1)rpw_2}{(N-1)(rp)^2 + rp + N-2} + \frac{(N-1)(rp+N-2)w_1}{(N-1)(rp)^2 + rp + N-2}.
\]
Adding \eqref{eq:star_11}, the mean absorption time for one single mutant is given by: 
\begin{align}
  \label{eq:ATstar}
 \tau_1 & = \frac{1}{N} \Big( \frac{w_1}{N-1+rp} + (N-1)^2 \frac{ rpw_2 + (rp+N-2)w_1}{(N-1)(rp)^2 + rp + N-2}\Big). 
\end{align}
More generally, we have: 
\begin{align} \label{eq:ATstar_i1}
\begin{split}
  \tau_{i,1} &= \hat p_{i,1}^- \tau_{i,0} 
                  + \hat p_{i,1}^+ \tau_{N-1,1} + 
			\frac{1}{p_{i,1}^- +p_{i,1}^+}   \\ 
             &= \frac{N-1-i}{N-1-i+rp}\tau_{i,0} 
                  + \frac{rp}{N-1-i+rp}\tau_{N-1,1} + \frac{w_{i+1}}{N-1-i+rp}   \\
             &= \frac{N-1-i}{N-1-i+rp}\tau_{i,0} + \frac{w_{i+1}}{N-1-i+rp},
\end{split}
\end{align}
\begin{align} \label{eq:ATstar_i0}
\begin{split}
\tau_{i,0} & = \hat p_{i,0}^- \tau_{i-1,0} 
                  + \hat p_{i,0}^+ \tau_{i,1} +  
	\frac{1}{p_{i,0}^- +p_{i,0}^+}   \\ 
             &= \frac{1}{(N-1)rp+1} \tau_{i-1,0} 
                  + \frac{(N-1)rp}{(N-1)rp+1}\tau_{i,1} + 
\frac{(N-1)w_i}{\big((N-1)rp+1\big)i}. 
\end{split}
\end{align}
Combining both equations, we obtain the following recurrence formula: 
\begin{align} \label{eq:ATrecurrent}
\begin{split}
\tau_{i,0} = \frac{\displaystyle 1}{\displaystyle
\Big[ 1 - \frac{(N-1)rp}{[(N-1)rp+1]}\frac{N-1-i}{[N-1-i+rp]} \Big]} 
&  \Big( \frac{1}{[(N-1)rp+1]} \tau_{i-1,0} \\
 & + \frac{(N-1)rp}{[(N-1)rp+1]}\frac{w_{i+1}}{[N-1-i+rp]}  \\
& + \frac{(N-1)}{[(N-1)rp+1]}\frac{w_i}{i}\Big). 
\end{split}
\end{align}
The absorption time for one single mutant in the Moran process has been described in \cite{H}: 
\begin{align} \label{eq:ATMoranstar}
\begin{split}
 \tau_1^M  & = \frac{1}{N} \Big( \big( \frac{r}{N-1+r} + (N-1)^2 \frac{r}{(N-1)r+1} \big)\frac{1}{D(1)}  \sum_{i=2}^{N-1} D(i)E(i)  \\
& + 1 + (N-1)^2 \frac{w_1}{(N-1)r + 1} \Big), 
\end{split}
\end{align}
where
\begin{align*}
D(i)  &  = 1 + \frac{N-1}{N-1+r} \sum_{j=i}^{N-2}  \Big(\frac{N-1+r}{(N-1)r+1}\Big)^{j-i+1} \Big(\frac{1}{r}\Big)^{j-i+1}, \\
E(i)  & = \frac{N-1}{r} \sum_{j=1}^{i-1} \frac{(N-1)w_j}{((N-1)r + 1)j} \big(\frac{1}{(N-1)r+1}  \big)^{i-j-1}  + \frac{(N-1)w_i}{(N-i)r} \\ 
& = \frac{N-1}{r} \sum_{j=1}^{i-1} \frac{(N-1)(N-j+rj)}{((N-1)r + 1)j} \big(\frac{1}{(N-1)r+1}  \big)^{i-j-1}  + \frac{(N-1)(N-i+ri)}{(N-i)r}.
\end{align*}

\begin{proposition} \label{prop:AR_star}
The absorption time for a single mutant individual on a star graph satisfies the condition:  
$$
\lim_{r \to +\infty} \rho(r,p_c)  = \frac{N}{3N-2} 
\big[ (N-1) \sum_{i=1}^{N-1} H_ i \big].
$$
\end{proposition}

\begin{proof} To compute 
$\lim_{r \to +\infty} \rho(r,p)$,
we will compute independently the limits of $\tau_1$ and $\tau_1^M$ using 
equations \eqref{eq:ATstar} and \eqref{eq:ATMoranstar}. First, 

\begin{align*}
\lim_{r \to +\infty} \tau_1 = \frac 1 N \Big( \frac 1 p + (N-1)^2 \frac{2p + p}{(N-1)p^2} \Big) 
= \frac{3N-2}{Np}.
\end{align*}
Secondly, we compute the limit of each summand in \eqref{eq:ATMoranstar}. Since 
\begin{align*}
\lim_{r \to +\infty} E(i) = \frac{N-1}{N-i}i 
\end{align*}
and 
\begin{align*}
\lim_{r \to +\infty} \frac{D(i)}{D(1)}   = 
\lim_{r \to +\infty}  \frac{\displaystyle 1 + \frac{N-1}{N-1+r} \sum_{j=i}^{N-2}  \Big(\frac{N-1+r}{(N-1)r+1}\Big)^{j-i+1} \Big(\frac{1}{r}\Big)^{j-i+1}}{\displaystyle 1 + \frac{N-1}{N-1+r} \sum_{j=1}^{N-2}  \Big(\frac{N-1+r}{(N-1)r+1}\Big)^{j} \Big(\frac{1}{r}\Big)^{j}} = 1,
\end{align*}
the first summand converges to 
$$
\sum_{i=2}^{N-1} \frac{N-1}{N-i}i =  (N-1) \sum_{i=1}^{N-1} \frac{i}{N-i} - 1 = (N-1) \sum_{i=1}^{N-1} H_ i -1,
$$
where 
$\sum_{i=1}^{N-1} H_ i = NH_{N-1} - (N-1) = N(H_{N-1}-1) +1$.
The second one is equal to $1/N$, and finally the third summand converges to $\frac{N-1}{N}$. Therefore we have:
\begin{align*}
\qquad \lim_{r \to +\infty} \rho(r,p) = \lim_{r \to +\infty} \frac{\tau_1^M}{\tau_1} 
 = \frac{\displaystyle \sum_{ i=2}^{N-1} \frac{N-1}{N-i} i  + \frac 1 N + \frac{N-1}{N}}{\displaystyle \frac{3N-2}{Np}} 
= p \frac{N}{3N-2} \big[ (N-1)
\sum_{i=1}^{N-1} H_ i \big]. \hspace{3em} 
\end{align*}
Remember finally that $\lim_{r \rightarrow + \infty} p_{c}(r) = 1$. 
\end{proof} 

\begin{remark} \label{rem:redstar}
From the previous proof, we also deduce that 
$$\lim_{r \to +\infty} \tau_1^M = \lim_{r \to +\infty} T_1^M = 
 (N-1) \sum_{i=1}^{N-1} H_ i.$$
Together with Remarks~\ref{rem:redcomplete}~and~\ref{rem:redcycle}, this remark seems to suggest that the limit is related to the \emph{cover time} of the simple random walk (SRW in short) on the underlying graph (see \cite{Lovasz1993} and \cite{LevinPeres}). See Supplementary Material S3.
Both quantities coincide for complete and cycle graphs, although SRW is more efficient in star graphs.
\end{remark}

\subsection{Fixation time for Bernoulli proliferation on stars}

Using notations from the previous paragraph, we have:  
\begin{align}  \label{eq:FTstar_11}
\begin{split}
  z_{0,1} &= \hat p_{0,1}^- z_{0,0} 
                  + \hat p_{0,1}^+ z_{N-1,1} + 
			\frac{\Phi_{0,1}}{p_{0,1}^- +p_{0,1}^+} \phantom{X}  \\
             &= \frac{w_1\Phi_{0,1}}{N-1+rp},  
\end{split}
\end{align}
\begin{align} \label{eq:FTstar_12}
\begin{split}
z_{1,0} &= \hat p_{1,0}^- z_{0,0} 
                  + \hat p_{1,0}^+ z_{1,1} +  
	\frac{\Phi_{1,0}}{p_{1,0}^- +p_{1,0}^+}  \\
  &= \frac{(N-1)rp}{(N-1)rp+1}z_{1,1} + 
\frac{(N-1)w_1\Phi_{1,0}}{(N-1)rp+1}, 
\end{split}
\end{align}
\begin{align} \label{eq:FTstar_13} 
\begin{split}
z_{1,1} &= \hat p_{1,1}^- z_{1,0} + \hat p_{1,1}^+ z_{N-1,1}+ \frac{\Phi_{1,1}}{p_{1,1}^- +p_{1,1}^+}  \phantom{X} \\
       &= \frac{N-2}{N - 2 + rp} z_{1,0} +  \frac{w_2\Phi_{1,1}}{N-2+rp}, 
\end{split}
\end{align}
where $z_{i,\varepsilon} = T_{i,\varepsilon} \Phi_{i,\varepsilon}$ for any $i = 0,\dots, N-1$ and any $\varepsilon = 0,1$. 
From \eqref{eq:FTstar_12} and \eqref{eq:FTstar_13}, we deduce:
\[
 T_{1,0} = (N-1) \frac{ rpw_2\frac{\Phi_{1,1}}{\Phi_{1,0}} + (rp+N-2)w_1}{(N-1)(rp)^2 + rp + N-2}.
\]
Adding \eqref{eq:FTstar_11}, the mean fixation time for one single mutant is given by: 
\begin{equation}
  \label{eq:FT_star}
 T_1  = \frac{1}{N} \Big( \frac{w_1}{N-1+rp} + (N-1)^2 \frac{ rpw_2\frac{\Phi_{1,1}}{\Phi_{1,0}} + (rp+N-2)w_1}{(N-1)(rp)^2 + rp + N-2} \Big). 
\end{equation}
More generally, we have: 
\begin{equation} \label{eq:FTstar_i1}  
  z_{i,1} = \hat p_{i,1}^- z_{i,0} 
                  + \hat p_{i,1}^+ z_{N-1,1} + 
			\frac{\Phi_{i,1}}{p_{i,1}^- +p_{i,1}^+} 
             = \frac{N-1-i}{N-1-i+rp}z_{i,0} + \frac{w_{i+1}\Phi_{i,1}}{N-1-i+rp}, 
\end{equation}
\begin{align} 
\begin{split}
z_{i,0} &= \hat p_{i,0}^- z_{i-1,0} 
                  + \hat p_{i,0}^+ z_{i,1} +  
	\frac{\Phi_{i,0}}{p_{i,0}^- +p_{i,0}^+}  \label{eq:FTstar_i2} \\ 
             &= \frac{1}{(N-1)rp+1} z_{i-1,0} 
                  + \frac{(N-1)rp}{(N-1)rp+1}z_{i,1} + 
\frac{(N-1)w_i\Phi_{i,0}}{\big((N-1)rp+1\big)i}. \phantom{XXXXXXX}
\end{split}
\end{align}
Combining both equations, we obtain a new recurrence formula for the peripheral process: 
\begin{align} \label{eq:TFrecurrent}
T_{i,0} = \frac{\displaystyle 1}{\displaystyle
\Big[ 1 - \frac{(N-1)rp}{[(N-1)rp+1]}\frac{N-1-i}{[N-1-i+rp]} \Big]} 
&  \Big( \frac{1}{[(N-1)rp+1]} T_{i-1,0}  \frac{\Phi_{i-1,0}}{\Phi_{i,0}} \nonumber \\
 & + \frac{(N-1)rp}{[(N-1)rp+1]}\frac{w_{i+1}}{[N-1-i+rp]} \frac{\Phi_{i,1}}{\Phi_{i,0}}\\
& + \frac{(N-1)}{[(N-1)rp+1]}\frac{w_i}{i} \Big). \nonumber
\end{align}


\noindent
Once again, combining Proposition~\ref{prop:infinityfitness} and Proposition~\ref{prop:AR_star}, we have: 
$$\lim_{r \to +\infty} \Rho(r,p_c)  = \lim_{r \to +\infty} \rho(r,p_c)  = \frac{N}{3N-2} 
 \big[ (N-1) \sum_{i=1}^{N-1} H_ i \big].
$$
The behaviour of stars with respect to cliques and cycles can be also graphically visualised in Supplementary Material S5.

\section{Absorption and fixation times for binomial proliferation}
\label{Sbinomial}

In this section we shall describe the absorption and fixation times for binomial proliferation on a complete graph. We will combine the idea used to compute the fixation probability in \cite{AlcaldeGuerberoffLozano2022} with the technique used to compute the absorption and fixation times for the usual Moran process in \cite{H}. As before, we suppose the starting mutant is placed at vertex~1. 
It may evolve to any other state: the embedded jump processes in order $N=3,4$ are described in Figure~\ref{fig:binomial}.
\begin{figure}[b]
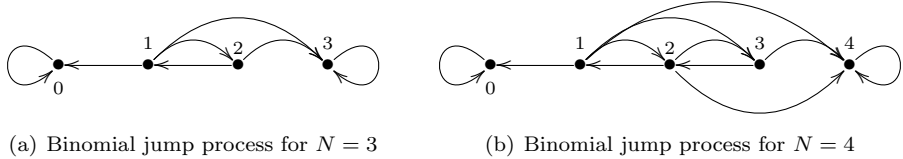

\centering
\subfigure[Binomial jump process for $N=3$]{
\label{fig:binomial3}
  $ \xy
    0;/r.14pc/: 
	(-40,-15)*{};
    (-80,0)*{\bullet}="A0"; 
    (-80,-5)*{\scriptstyle 0};
    (-60,0)*{\bullet}="B0"; 
    (-60,5)*{\scriptstyle 1};
    (-40,0)*{\bullet}="C0"; 
    (-40,4)*{\scriptstyle  2};
    (-20,0)*{\bullet}="D0"; 
   (-20,5)*{\scriptstyle  3};
    {\ar@/^0pc/ "B0"; "A0"}; 
    {\ar@/^0.8pc/ "B0"; "C0"}; 
    {\ar@/^1.5pc/ "B0"; "D0"}; 
    {\ar@/^0pc/ "C0"; "B0"}; 
    {\ar@/^0.8pc/ "C0"; "D0"};  
    {\ar@(lu,ld) "A0"; "A0"};   
    {\ar@(ru,rd) "D0"; "D0"};  
    \endxy
    $ 
}
\subfigure[Binomial jump process for $N=4$]{
\label{fig:binomial4}
  $ \xy
    0;/r.14pc/:
	(-40,-15)*{};
    (-80,0)*{\bullet}="A0"; 
    (-80,-5)*{\scriptstyle 0};
    (-60,0)*{\bullet}="B0"; 
    (-60,5)*{\scriptstyle 1};
    (-40,0)*{\bullet}="C0"; 
    (-40,4)*{\scriptstyle  2};
    (-20,0)*{\bullet}="D0"; 
   (-20,5)*{\scriptstyle  3};
    (0,0)*{\bullet}="E0"; 
   (0,5)*{\scriptstyle  4};
    {\ar@/^0pc/ "B0"; "A0"}; 
    {\ar@/^0.8pc/ "B0"; "C0"}; 
    {\ar@/^1.5pc/ "B0"; "D0"}; 
    {\ar@/^2pc/ "B0"; "E0"}; 
    {\ar@/^0pc/ "C0"; "B0"}; 
    {\ar@/^0.8pc/ "C0"; "D0"}; 
    {\ar@/^0pc/ "D0"; "C0"};  
    {\ar@/^0.8pc/ "D0"; "E0"};  
    {\ar@/_1.5pc/ "C0"; "E0"}; 
    {\ar@(lu,ld) "A0"; "A0"}; 
    {\ar@(ru,rd) "E0"; "E0"};  
    \endxy
    $
}
  \caption{Binomial jump proliferation on a complete graph of \subref{fig:binomial3}~order $N=3$, \subref{fig:binomial4}~order $N=4$.}
\label{fig:binomial}
\end{figure}

The absorption time is given as a solution of the non-homogeneous linear system
\begin{equation}
  \label{eq:Asistema1}
	\tau_i = p_{i,i-1} \tau_{i-1} + p_{i,i} \tau_i + 
    \cdots + p_{i,N} \tau_N + 1,
\end{equation}
where $i = 1,2,\ldots,N-1$. If $j \geq i$, the coefficient are
$$
  p_{i,j}
  = \frac{ri}{ri+N-i} \binom{N-i}{j-i} p^{j-i} (1-p)^{N-j},
$$
while
$$
  p_{i,i-1} = \frac{N-i}{ri+N-i}\,\frac{i}{N-1}. 
$$
We write $\Delta_{i-1,i} = \tau_{i-1} - \tau_i$ and $\Delta_{i,j} = \tau_{i} - \tau_{j}$ for $j \geq i+1$. Thus, since $\sum_{j =i-1}^N p_ {i,j} = 1$,
we deduce:
\begin{equation}
  \label{eq:ADeltacomplete1}
  p_{i,i-1} \Delta_{i-1,i} = p_{i,i+1}\Delta_{i,i+1} + 
 \cdots + p_{i,N}\Delta_{i,N} - 1.
\end{equation}
As
$\Delta_{i,j} = \Delta_{j-1,j} + \Delta_{j-2,j-1} + \cdots + \Delta_{i,i+1}$, 
we can reorder the terms of \eqref{eq:ADeltacomplete1} to prove:
\begin{align*}
p_{i,i-1} \delta_i =~ &  p_{i,i+1}\delta_{i+1} +  \\
 & p_{i,i+2}\delta_{i+1}  + p_{i,i+2}\delta_{i+2} + \\
 & p_{i,i+3}\delta_{i+1}  + p_{i,i+3}\delta_{i+2} + 
p_{i,i+3} \delta_{i+3} + \cdots \\
 & p_{i,N} \phantom{x} \delta_{i+1}  + p_{i,N}\phantom{x} \delta_{i+2}  + p_{i,N} \phantom{x}\delta_{i+3} + \cdots + p_{i,N}\delta_N - 1,
\end{align*}
where $\delta_i = \Delta_{i-1,i}$ and $\delta_{j} = \Delta_{j-1,j}$ for $j \geq i+1$.
Hence, defining 
$$
w_{i,j} = \frac1{ p_{i,i-1} }\sum_{k=j}^N p_{i,k},
$$
we reduce \eqref{eq:Asistema1} to the linear system
\begin{equation}
  \label{eq:Asistema2}
  \delta_i = w_{i,i+1} \delta_{i+1} + 
    w_{i,i+2} \delta_{i+2} + \cdots + w_{i,N-1} \delta_{N-1} + w_{i,N} \delta_{N} - \frac{1}{p_{i,i-1}},
\end{equation}
where $i = 1,2,\ldots,N-1$. This system can be solved by backward substitution. 
Firstly, we take 
$$\delta_{N-1} = w_{N-1,N}\delta_N -  \frac{1}{p_{N-1,N-2}}.$$ Substituting in the expression for $\delta_{N-2}$, we get:
\begin{align*}
  \delta_{N-2} & = w_{N-2,N-1} \delta_{N-1} + w_{N-2,N}\delta_N -  \frac{1}{p_{N-2,N-3}} \\ 
 & = \bigl(w_{N-2,N-1} w_{N-1,N} + w_{N-2,N}\bigr)\delta_N - \frac{w_{N-2,N-1}}{p_{N-1,N-2}} -  \frac{1}{p_{N-2,N-3}},
\end{align*}
and continue recursively. 
Notice all possible strictly increasing sequences of numbers from $N-2$ to
$N$ appears in the term involving $\delta_N $, while those from $N-2$ to
$N-1$ \mbox{(with $w_{N-2,N-2} = 1$)} appears in the constant terms. In general, let $\mathcal{S}_{i,j}$ be the set of strictly increasing sequences (with variable length $\ell$) $s \equiv 
s_0 < s_1 < \ldots < s_\ell$ 
of integers such that $s_0=i$ and $s_\ell=j$ for $j \geq i$, with the convention $w_{j,j} = 1$ for all $j$, and denote
$W(s) = w_{s_0,s_1}w_{s_1,s_2}\cdots w_{s_{\ell-1}, s_\ell}$.
We have:
\begin{equation}
  \label{eq:Adeltafinal}
\delta_i = \delta_N  \sum_{s\in\mathcal{S}_{i,N}} W(s)  - \sum_{j=i}^{N-1} \frac{1}{p_{j,j-1}} \sum_{s\in\mathcal{S}_{i,j}} W(s).
\end{equation}
Arguing similarly to \cite{AlcaldeGuerberoffLozano2022}, as $\sum_{i=1}^N \delta_i = 0$, we obtain:\begin{equation} \label{eq:deltaN}
  \delta_N = \frac{\sum_{i=1}^{N-1} \sum_{j=i}^{N-1}  \frac{1}{p_{j,j-1}} \sum_{s\in\mathcal{S}_{i,j}} W(s)}{\sum_{i=1}^{N} \sum_{s\in\mathcal{S}_{i,N}} W(s)}.
\end{equation}
Finally, using \eqref{eq:Adeltafinal} and the equation above, we get:
\begin{align} \label{eq:ATbinomial}
\begin{split}
  \tau_1\; = \; - \delta_1 \; & = \sum_{j=1}^{N-1} \frac{1}{p_{j,j-1}} \sum_{s\in\mathcal{S}_{1,j}} W( s) \\  & -  \sum_{\tilde{s} \in \mathcal{S}_{1,N}} W(\tilde{s}) \frac{\sum_{j=1}^{N-1}  \sum_{k=j}^{N-1} \frac{1}{p_{k,k-1}} \sum_{s\in\mathcal{S}_{j,k}}  W(s)}{ \sum_{k=1}^{N}  \sum_{s\in\mathcal{S}_{k,N}} W(s)}.
\end{split}
\end{align}
Also recall 
\[
  w_{i,j} = r\frac{(N-1)}{N-i}
    \sum_{k=j}^N \binom{N-i}{k-i} p^{k-i} (1-p)^{N-k}.
\]

The same approach can be used to compute the absorption time for the Moran process (or any Birth-Death process) on a complete graph. In this case, the weight 
$w_{i,j} = p^M_{i,i+1}/p^M_{i,i-1} =  1/\gamma_i$ 
if $j = i+1$ and $w_{i,j}=0$ otherwise, 
and \eqref{eq:Asistema2} reduces to
 \begin{equation}
  \label{eq:Asistema3}
  \delta_i = w_{i,i+1} \delta_{i+1} - \frac{1}{p^M_{i,i-1}}
\end{equation}
for $i = 1,2, \dots , N-1$. Arguing as before, we obtain the following expression:
\begin{align} \label{eq:ATBD1}
  \tau^M_1  & = \sum_{i=1}^{N-1} \frac{1}{p^M_{i,i-1}} 
\prod_{k=1}^{i-1} w_{k,k+1}  -  \frac{\displaystyle \prod_{l=1}^{N-1} w_{l,l+1} \sum_{i=1}^{N-1}  \sum_{j=i}^{N-1} \frac{1}{p^M_{j,j-1}} \prod_{k=i}^{j-1} w_{k,k-1}}{\displaystyle 1 +  \sum_{i=1}^{N-1}  \sum_{k=i}^{N-1} w_{k,k+1}}. 
\end{align}
However, replacing $\delta_i$ by $-\delta_i$, the equation \eqref{eq:Asistema3} can be replaced by 
 \begin{equation}
  \label{eq:Asistema4}
  \delta_{i+1} = \gamma_i \delta_i - \frac{1}{p^M_{i,i+1}}, 
\end{equation}
and now we obtain the following equivalent expression (cf. \cite[Section 4.7]{KT}): 
\begin{align} \label{eq:ATBD2}
  \tau^M_1 &  = \frac{\displaystyle \sum_{i=1}^{N-1}  \frac{1}{p^M_{i,i+1}}\prod_{k=1}^{i} w_{k,k+1} \sum_{j=i}^{N-1}  \prod_{k=j+1}^{N-1} w_{k,k+1}}{\displaystyle \sum_{i=0}^{N-1} \prod_{k=i+1}^{N-1} w_{k,k+1}}.
\end{align}

 The fixation time is now the solution of the non-homogeneous linear system
\begin{equation}
  \label{eq:Fsistema1}
	\Phi_i T_i = p_{i,i-1} \Phi_{i-1} T_{i-1}+ p_{i,i} \Phi_i  T_i+ 
    \cdots + p_{i,N} \Phi_N T_N+ \Phi_i,
\end{equation}
where $\Phi_N = 1$ and $T_N = 0$. As in the previous computations, writing 
$z_i = \Phi_i  T_i$, we reduce \eqref{eq:Fsistema1} to the linear system 
\begin{equation}
  \label{eq:Fsistema2}
	z_i = p_{i,i-1} z_{i-1} + p_{i,i} z_i + 
    \cdots + p_{i,N} z_N + \Phi_i
\end{equation}
similar to \eqref{eq:Asistema1}. This system can be solved using the same method. 

\section{Symmetry and monotonicity}
\label{Smonotone}

The aim of this section is to prove  Theorem~\ref{thm:decreaseMoran} for the Moran process $\Proc^M$ and Theorem~\ref{thm:decreasebeta} for the proliferation processes $\Proc^\beta$ with $\beta = B,b$. Let $\Proc$ be any process $\Proc^M$ or $\Proc^\beta$ defined on a connected undirected graph $G = (V,E)$ with fitness $r \in (0,+\infty)$ and proliferation parameter $p \in (0,1]$. The process $\Proc$ identifies to the random walk on the state space $\cals = \calp(V)$ with the graph structure given by the transition matrix $P$. Two states $S, S' \in \cals$ are joined by an oriented edge if and only if $p_{S,S'} > 0$ and this probability becomes the weight of the edge. We call \emph{symmetry} of $\Proc$ any symmetry of the weighted and directed graph $\cals$. Denoting $\calp$ the space of trajectories $\sigma = \{S_n\}_{n \geq 0}$ and $\calp_{S_0}$ the subspace of trajectories starting at $S_0$, the probability measure on $\calp_{S_0}$ is given by
$$
\Pp_{S_0} [\sigma(0) = S_0, \dots, \sigma(n) = S_n] = p_{S_0,S_1} \cdots p_{S_{n-1},S_n}.
$$
A more complete description can be seen in Supplementary Material S1 and S2. 

\subsection{G-symmetry.} Monotonicity of the expected fixation time in the Moran process will be based on some kind of symmetry. Let us restrict to the classical case $\Proc = \Proc^M$ 
on the complete graph $G = K_N$. Any state $S \in \calp(V)$ is a clique (of graph $G$) with $i=|S|$ elements endowed with the subgraph structure. Thus the symmetries of $G$ become symmetries of $\Proc$ reducing it to the random walk on the quotient graph $\calk_N$, which is described in Figure~\ref{fig:reducedcomplete} below. 
\begin{figure}[h!]
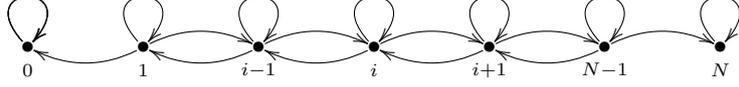

\centering
 $ \xy
    0;/r.12pc/: 
    (-90,0)*{\bullet}="A0"; 
    (-90,-6)*{\scriptstyle 0};
    (-60,0)*{\bullet}="B0"; 
    (-60,-6)*{\scriptstyle 1};
    (-30,0)*{\bullet}="C0"; 
    (-30,-6)*{\scriptstyle  i-1};
    (0,0)*{\bullet}="D0"; 
    (0,-6)*{\scriptstyle  i};
    (30,0)*{\bullet}="E0"; 
    (30,-6)*{\scriptstyle i+1};
    (60,0)*{\bullet}="F0"; 
    (60,-6)*{\scriptstyle  N-1};
    (90,0)*{\bullet}="G0"; 
    (90,-6)*{\scriptstyle N};
    (0,-12)*{};
    {\ar@/^0.5pc/ "B0"; "A0"}; 
    {\ar@/^0.5pc/ "B0"; "C0"}; 
    {\ar@/^0.5pc/ "C0"; "B0"}; 
    {\ar@/^0.5pc/ "C0"; "D0"}; 
    {\ar@/^0.5pc/ "D0"; "C0"}; 
    {\ar@/^0.5pc/ "D0"; "E0"}; 
    {\ar@/^0.5pc/ "E0"; "D0"}; 
    {\ar@/^0.5pc/ "E0"; "F0"}; 
    {\ar@/^0.5pc/ "F0"; "E0"}; 
    {\ar@/^0.5pc/ "F0"; "G0"};
    {\ar@(ul,ur) "A0"; "A0"}; 
    {\ar@(ul,ur) "A0"; "A0"}; 
    {\ar@(ul,ur) "B0"; "B0"}; 
    {\ar@(ul,ur) "C0"; "C0"}; 
    {\ar@(ul,ur) "D0"; "D0"}; 
    {\ar@(ul,ur) "E0"; "E0"}; 
    {\ar@(ul,ur) "F0"; "F0"}; 
    {\ar@(ul,ur) "A0"; "A0"}; 
    {\ar@(ul,ur) "G0"; "G0"}
    \endxy
    $
\caption{The reduced state space $\calk_N$ for the Moran process $\Proc^M$.}
\label{fig:reducedcomplete}
\end{figure}

The fixation probability $\Phi_{S_0}$ is equal to the fixation probability $\Phi_{|S_0|}$ of the reduced model. We also know that $\Proc^M$ can be replaced by the embedded jump process $\hat \Proc^M$ (see Supplementary Material S2). It is easy to see that the reduced process of $\hat \Proc^M$ is the embedded jump process of the reduced process of $\Proc^M$ with forward probability $\frac{r}{r+1}$ and backward probability $\frac{1}{r+1}$. Its state space $\hat \calk_N$ is also shown in Figure~\ref{fig:jumpcomplete}. 
\begin{figure}[h!]
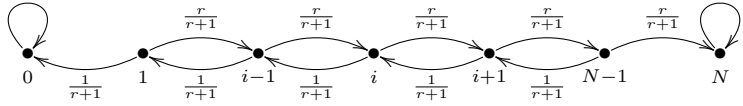

\centering
 $ \xy
    0;/r.12pc/: 
    (-90,0)*{\bullet}="A0"; 
    (-90,-6)*{\scriptstyle 0};
    (-60,0)*{\bullet}="B0"; 
    (-60,-6)*{\scriptstyle 1};
    (-30,0)*{\bullet}="C0"; 
    (-30,-6)*{\scriptstyle  i-1};
    (0,0)*{\bullet}="D0"; 
    (0,-6)*{\scriptstyle  i};
    (30,0)*{\bullet}="E0"; 
    (30,-6)*{\scriptstyle i+1};
    (60,0)*{\bullet}="F0"; 
    (60,-6)*{\scriptstyle  N-1};
    (90,0)*{\bullet}="G0"; 
    (90,-6)*{\scriptstyle N};
    (0,-14)*{};
    (-75,-9)*{\scriptstyle \frac{1}{r+1}};
    (-45,-9)*{\scriptstyle \frac{1}{r+1}};
    (-15,-9)*{\scriptstyle \frac{1}{r+1}};
    (15,-9)*{\scriptstyle \frac{1}{r+1}};
    (45,-9)*{\scriptstyle \frac{1}{r+1}};
    (75,9)*{\scriptstyle \frac{r}{r+1}};
    (-45,9)*{\scriptstyle \frac{r}{r+1}};
    (-15,9)*{\scriptstyle \frac{r}{r+1}};
    (15,9)*{\scriptstyle \frac{r}{r+1}};
    (45,9)*{\scriptstyle \frac{r}{r+1}};
    {\ar@/^0.5pc/ "B0"; "A0"}; 
    {\ar@/^0.5pc/ "B0"; "C0"}; 
    {\ar@/^0.5pc/ "C0"; "B0"}; 
    {\ar@/^0.5pc/ "C0"; "D0"}; 
    {\ar@/^0.5pc/ "D0"; "C0"}; 
    {\ar@/^0.5pc/ "D0"; "E0"}; 
    {\ar@/^0.5pc/ "E0"; "D0"}; 
    {\ar@/^0.5pc/ "E0"; "F0"}; 
    {\ar@/^0.5pc/ "F0"; "E0"}; 
    {\ar@/^0.5pc/ "F0"; "G0"};
    {\ar@(ul,ur) "A0"; "A0"}; 
    {\ar@(ul,ur) "G0"; "G0"}
    \endxy
    $
\caption{The state space $\hat \calk_N$ for the embedded jump process $\hat \Proc^M$ on a complete graph $G=K_N$.}
\label{fig:jumpcomplete}
\end{figure}

However, to compute the expected fixation time, some care must be taken. If the number of non-trivial transitions of a trajectory $\sigma \in \calp_{S_0}$ is equal to $n$, then 
the length of $\sigma$ (in average) is given by
\begin{equation} \label{eq:length}
\tau_{S_0}(\sigma) = n + \sum_{i=0}^{n-1} \frac{p_{S_i,S_i}}{1 - p_{S_i,S_i}} 
= n + \sum_{i=0}^{n-1} \big( \frac{1}{1 - p_{S_i,S_i}} -1 \big)  = \sum_{i=0}^{n-1} \frac{1}{1 - p_{S_i,S_i}}
\end{equation}
where the sojourn time in each state $S_i$ does not depend on $S_i$, but only on $|S_i|$.

Now we can observe that the weighted graph $\calk_N$ has a global symmetry, which we call \emph{G-symmetry}, when we replace $r$ with $1/r$. Equivalently this symmetry applies to the weighted graph $\hat \calk_N$ and the probabilities 
\begin{equation} \label{eq:sojourn}
p_{i,i} = p_{i,i}(r) = \frac{ri}{ri+N-i}.\frac{i-1}{N-1} + \frac{N-i}{ri+N-i}.\frac{N-i-1}{N-1}
\end{equation} 
(or in the same way the sojourn times $1/1-p_{i,i}$) decorating the vertices of $\hat \calk_N$. Namely $p_{i,i}(\frac 1 r) = p_{N-i,N-i}(r)$ for all $i = 1, \dots , N-1$ and all $r > 0$.

As we explain in Supplementary Material S2, conditioning to the fixation event is equivalent to take the \emph{Doob transform} to the process $\Proc$. A detailed definition can be found in \cite[Chapter 17]{LevinPeres} and also in the above cited  Supplementary Material. From Proposition S1, the Doob transfom of the embedded jump process is the same that the embedded jump process of the Doob transform. Now, using G-symmetry, we retrieve the symmetry between the expected fixation time for $i$ mutant individuals and the expected extinction time for $N-i$ resident individuals, formulated by Maruyama and Kimura in \cite{MaruyamaTimes1974, MaruyamaKimura} (see also \cite{Antal2006, Traulsen2015, Taylor2006}). We also see that the expected fixation time of mutants of fitness $r$ and $1/r$ are identical, an observation which is derived from the Maruyama-Kimura symmetry in \cite{HathcockStrogatz}. However, we should note that the results from \cite{Traulsen2015} and \cite{HathcockStrogatz} are more general and involve the whole fixation time distribution. 

The same construction applies when $G$ is a cycle graph, exhibiting a global G-symmetry for all essential states in the sense of \cite{Seneta} (see Figure~\ref{fig:reducedcycle}). They form the communicating class of $1$. But below we describe two examples where this G-symmetry is broken.
\begin{figure}[h!]
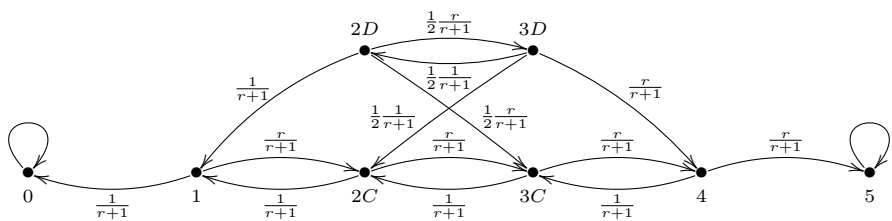
 
\centering
$
\xy
0;/r.175pc/: 
(-60,0)*{\bullet}="A0"; 
(-60,-4)*{\scriptstyle 0};
(-30,0)*{\bullet}="B0"; 
(-30,-4)*{\scriptstyle 1};
(0,0)*{\bullet}="C0"; 
(0,-4)*{\scriptstyle 2C};
(0,22)*{\bullet}="C1"; 
(0,26)*{\scriptstyle 2D};
(30,0)*{\bullet}="D0"; 
(30,-4)*{\scriptstyle 3C};
(30,22)*{\bullet}="D1"; 
(30,26)*{\scriptstyle 3D};
(60,0)*{\bullet}="E0";
(60,-4)*{\scriptstyle 4};
(90,0)*{\bullet}="F0";
(90,-4)*{\scriptstyle 5};
{\ar@/^0.5pc/  "B0"; "A0"}; 
{\ar@/^0.5pc/ "B0"; "C0"}; 
{\ar@/^0.5pc/ "C0"; "B0"}; 
{\ar@/^0.5pc/  "C0"; "D0"}; 
{\ar@/^0.5pc/  "D0"; "C0"};
{\ar@/_0.5pc/  "C1"; "B0"}; 
{\ar@/^0.4pc/  "C1"; "D1"};
{\ar@/^0.1pc/  "C1"; "D0"};
{\ar@/^0.4pc/  "D1"; "C1"}; 
{\ar@/_0.1pc/  "D1"; "C0"};
{\ar@/^0.4pc/  "D1"; "E0"}; 
{\ar@/^0.5pc/  "D0"; "E0"};
{\ar@/^0.5pc/  "E0"; "D0"}; 
{\ar@/^0.5pc/ "E0"; "F0"};  
{\ar@(ul,ur) "A0"; "A0"}; 
{\ar@(ul,ur) "F0"; "F0"}; 
(-45,-6)*{\scriptstyle \frac{1}{r+1}};
(-15,-6)*{\scriptstyle \frac{1}{r+1}};
(-15,5.5)*{\scriptstyle \frac{r}{r+1}};
(15,-6)*{\scriptstyle \frac{1}{r+1}};
(15,5.5)*{\scriptstyle \frac{r}{r+1}};
(15,27)*{\scriptstyle \frac 1 2 \! \frac{r}{r+1}};
(15,17)*{\scriptstyle \frac 1 2 \! \frac{1}{r+1}};
(45,-6)*{\scriptstyle \frac{1}{r+1}};
(45,5.5)*{\scriptstyle \frac{r}{r+1}};
(75,5.5)*{\scriptstyle \frac{r}{r+1}};
(-20,15)*{\scriptstyle \frac{1}{r+1}};
(50,15)*{\scriptstyle \frac{r}{r+1}};
(25,10)*{\scriptstyle \frac 1 2 \! \frac{r}{r+1}};
(5,10)*{\scriptstyle \frac 1 2 \! \frac{1}{r+1}};
\endxy
$
\caption{The reduced state space of $\hat \Proc^M$ on the cycle of order $N=5$. States marked with $C$ are connected and essential, and those marked with $D$ are disconnected and inessential.} \label{fig:reducedcycle}
\end{figure}

\subsection{Two examples} Let $L_3$ be the line graph with $3$ vertices, that is, the star graph $K_{1,2}$. Keeping in mind the notations used in Section~\ref{SBernoulli}, it is easy to describe the reduced state space $\hat \cals$ of the embedded jump process $\hat \Proc^M$, as drawn in Figure~\ref{fig:line} below. 
\begin{figure}[h!]
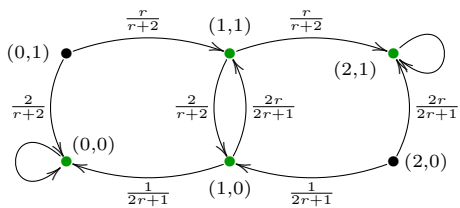

\centering
    $
     \xy
      0;/r.17pc/: 
      (-70,0)*{\green \bullet}="A0"; 
      (-65,3)*{\scriptstyle (0,0)};
      (-40,0)*{\green \bullet}="B0"; 
      (-40,-5)*{\scriptstyle (1,0)};
      (-10,0)*{\bullet}="C0";  
      (-4,0)*{\scriptstyle (2,0)};
      (-70,20)*{\bullet}="A1"; 
      (-77,20)*{\scriptstyle (0,1)};
      (-40,20)*{\green \bullet}="B1"; 
      (-40,25)*{\scriptstyle (1,1)};
      (-10,20)*{\green \bullet}="C1"; 
     (-17,17)*{\scriptstyle (2,1)};
      {\ar@/^0.5pc/  "B0"; "A0"}; 
      {\ar@/^0.5pc/ "C0"; "B0"}; 
      {\ar@/^0.5pc/ "A1"; "B1"}; 
      {\ar@/^0.5pc/  "B1"; "C1"};  
      {\ar@/_0.5pc/ "A1"; "A0"}; 
      {\ar@/_0.5pc/  "B1"; "B0"};
      {\ar@/_0.5pc/  "B0"; "B1"}; 
      {\ar@/_0.5pc/ "C0"; "C1"};        
      {\ar@(ru,rd) "C1"; "C1"};  
      {\ar@(lu,ld)  "A0"; "A0"}; 
      (-55,-6)*{\scriptstyle \frac{1}{2r+1}};
      (-25,-6)*{\scriptstyle \frac{1}{2r+1}};
      (-56,26)*{\scriptstyle \frac{r}{r+2}};
      (-26,26)*{\scriptstyle \frac{r}{r+2}};
      (-77,10)*{\scriptstyle \frac{2}{r+2}};   
      (-47,10)*{\scriptstyle \frac{2}{r+2}};
      (-32,10)*{\scriptstyle \frac{2r}{2r+1}};
      (-2,10)*{\scriptstyle \frac{2r}{2r+1}};
    \endxy
  $
  \caption{The reduced state space of $\hat \Proc^M$ on the line $L_3$.}
\label{fig:line}
\end{figure}

Obviously the central vertex and the two peripheral vertices show different behaviour and G-symmetry does not even make sense. Indeed, although the transition and sojourn probabilities are symmetric, the complementary states $(0,1)$ and $(2,0)$ do not belong to the same communicating class. However the subgraph of $\hat \cals$ determined by the green-coloured states admits a global G-symmetry since the sojourn probabilities 
$\frac 3 2 . \frac{1}{r+2}$ and $\frac 3 2 . \frac{r}{2r+1}$ 
at the states $(1,0)$ and $(1,1)$ are also exchanged when we replace $r$ with $1/r$. Thus the peripheral vertices of $L_3$ have the same expected fixation time, but the mean fixation time is not preserved by this symmetry.

Let us next consider the octahedron $G = O_6$, that is, the 4-regular graph of order $6$. According to the Circulation Theorem of \cite{LHN}, the Moran process on $O_6$ is equivalent to the  Moran process on $K_6$. This means that the number of elements $|S|$ of states $S \in \cals$ performs a random walk on the integer interval $[0,N]$ with absorbing states $0$ and $N$ and forward bias $r$. Indeed, using the symmetries of $G$, the state space of $\hat \Proc^M$ could be reduced to the graph $\hat \calk_6$ (see Figure~\ref{fig:jumpcomplete}). However this argument is only valid for the computation of the fixation probability since the symmetries of $G$ are not symmetries of $\hat \cals$. Distinguishing four central vertices from two peripheral vertices (as before for the star $K_{1,2}$), we obtain an intermediate state space $\hat \cals_1$ for $\hat \Proc^M$ drawn in Figure~\ref{fig:octahedron1}. 
\begin{figure}[h!]
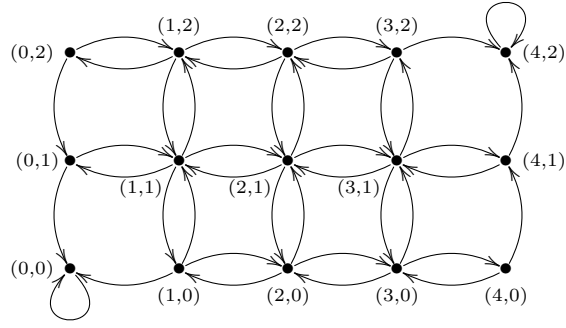

\centering
$
\xy
0;/r.17pc/: 
(-60,0)*{\bullet}="A0"; 
(-67,0)*{\scriptstyle (0,0)};
(-40,0)*{\bullet}="B0"; 
(-40,-5)*{\scriptstyle (1,0)};
(-20,0)*{\bullet}="C0"; 
(-20,-5)*{\scriptstyle (2,0)};
(0,0)*{\bullet}="D0"; 
(0,-5)*{\scriptstyle (3,0)};
(20,0)*{\bullet}="E0"; 
(20,-5)*{\scriptstyle (4,0)};
(-60,20)*{\bullet}="A1"; 
(-66,20)*{\scriptstyle (0,1)};
(-40,20)*{\bullet}="B1"; 
(-47,15)*{\scriptstyle (1,1)};
(-20,20)*{\bullet}="C1"; 
(-27,15)*{\scriptstyle (2,1)};
(0,20)*{\bullet}="D1"; 
(-7,15)*{\scriptstyle (3,1)};
(20,20)*{\bullet}="E1"; 
(27,20)*{\scriptstyle (4,1)};
(-60,40)*{\bullet}="A2"; 
(-67,40)*{\scriptstyle (0,2)};
(-40,40)*{\bullet}="B2"; 
(-40,45)*{\scriptstyle (1,2)};
(-20,40)*{\bullet}="C2"; 
(-20,45)*{\scriptstyle (2,2)};
(0,40)*{\bullet}="D2"; 
(0,45)*{\scriptstyle (3,2)};
(20,40)*{\bullet}="E2"; 
(27,40)*{\scriptstyle (4,2)};
{\ar@/^0.5pc/  "B0"; "A0"}; 
{\ar@/^0.5pc/ "C0"; "B0"}; 
{\ar@/^0.5pc/ "B0"; "C0"}; 
{\ar@/^0.5pc/  "D0"; "C0"}; 
{\ar@/^0.5pc/ "C0"; "D0"}; 
{\ar@/^0.5pc/ "E0"; "D0"}; 
{\ar@/^0.5pc/ "D0"; "E0"};  
{\ar@/^0.5pc/  "A1"; "B1"}; 
{\ar@/^0.5pc/  "B1"; "C1"}; 
{\ar@/^0.5pc/  "C1"; "D1"}; 
{\ar@/^0.5pc/  "D1"; "E1"};  
{\ar@/^0.5pc/  "B1"; "A1"}; 
{\ar@/^0.5pc/  "C1"; "B1"}; 
{\ar@/^0.5pc/  "D1"; "C1"}; 
{\ar@/^0.5pc/  "E1"; "D1"}; 
{\ar@/^0.5pc/  "A2"; "B2"}; 
{\ar@/^0.5pc/  "B2"; "A2"}; 
{\ar@/^0.5pc/  "B2"; "C2"}; 
{\ar@/^0.5pc/  "C2"; "B2"}; 
{\ar@/^0.5pc/ "C2"; "D2"}; 
{\ar@/^0.5pc/  "D2"; "C2"}; 
{\ar@/^0.5pc/  "D2"; "E2"}; 
{\ar@/_0.5pc/ "A1"; "A0"}; 
{\ar@/_0.5pc/ "B1"; "B0"}; 
{\ar@/_0.5pc/ "B0"; "B1"}; 
{\ar@/_0.5pc/ "C1"; "C0"}; 
{\ar@/_0.5pc/ "C0"; "C1"}; 
{\ar@/_0.5pc/ "D1"; "D0"}; 
{\ar@/_0.5pc/ "D0"; "D1"}; 
{\ar@/_0.5pc/ "E0"; "E1"}; 
{\ar@/_0.5pc/ "A2"; "A1"}; 
{\ar@/_0.5pc/ "B2"; "B1"}; 
{\ar@/_0.5pc/ "B1"; "B2"}; 
{\ar@/_0.5pc/ "C2"; "C1"}; 
{\ar@/_0.5pc/  "C1"; "C2"}; 
{\ar@/_0.5pc/ "D2"; "D1"}; 
{\ar@/_0.5pc/  "D1"; "D2"}; 
{\ar@/_0.5pc/ "E1"; "E2"}; 
{\ar@(dr,dl) "A0"; "A0"}; 
{\ar@(ul,ur) "E2"; "E2"}; 
\endxy
$
\caption{A reduction of the state space of $\hat \Proc^M$ on the octahedron $G = O_6$.}
\label{fig:octahedron1}
\end{figure}

The symmetries of $G$ allow us to identify all states in $\hat \cals$ or $\hat \cals_1$ with a single vertex. 
The same happens for all states with five vertices in $\hat \cals$ or $\hat \cals_1$. However, for states with 2 vertices, we retrieve the same connectedness phenomenon as for the cycle, although reduced states $2C$ and $2D$ are both essential. For states with 3 vertices, we can find two different subgraph structures induced by the graph structure of $G$, the line and the cycle, that we will denote $3L$ and $3T$. For states with 4 vertices, two different subgraph structures are also found: two triangles with a common face $4T$ and a cycle $4C$. As mentioned above, for states with $5$ vertices, only a pyramidal structure is induced. 
For computing the fixation probability of any state, the only thing that matters is its number of elements. However the subgraph structure is important to computing its expected fixation time. The sojourn time in state $3T$ is higher than in $3L$ since their probabilities are related by
$p_{3T,3T} = \frac 3 2 p_{3L,3L}$. The reduced state space is shown is Figure~\ref{fig:octahedron2} below.  
Observe that $3T$ evolves forward necessarily to $4T$, whereas $3L$ evolves forward to $4C$ with probability $1/3$ and to $4T$ with probability $2/3$. In the other direction $4T$ evolves backward equiprobably to $3L$ or $3T$.
In summary, there is no G-symmetry since we need to distinguish the different subgraph structures induced in states of order $3$ and $4$. 

\begin{figure}[h!]
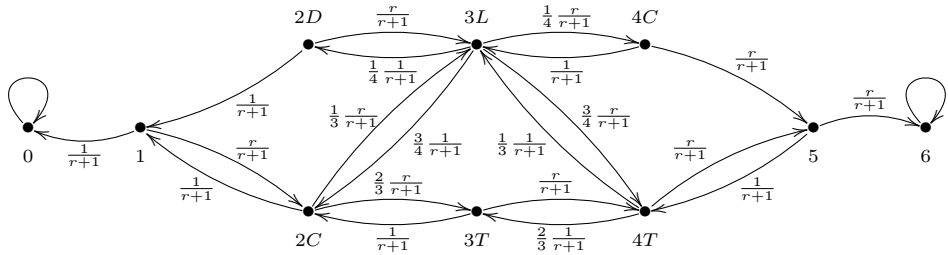
 
$
\xy
0;/r.175pc/: 
(-70,0)*{\bullet}="A0"; 
(-70,-5)*{\scriptstyle 0};
(-50,0)*{\bullet}="B0"; 
(-50,-5)*{\scriptstyle 1};
(-20,-15)*{\bullet}="C0"; 
(-20,-20)*{\scriptstyle 2C};
(-20,15)*{\bullet}="C1"; 
(-20,20)*{\scriptstyle 2D};
(10,-15)*{\bullet}="D0"; 
(10,-20)*{\scriptstyle 3T};
(10,15)*{\bullet}="D1"; 
(10,20)*{\scriptstyle 3L};
(40,-15)*{\bullet}="E0"; 
(40,-20)*{\scriptstyle 4T};
(40,15)*{\bullet}="E1"; 
(40,20)*{\scriptstyle 4C};
(70,0)*{\bullet}="F0";
(70,-5)*{\scriptstyle 5};
(90,0)*{\bullet}="G0"; 
(90,-5)*{\scriptstyle 6};
{\ar@/^0.4pc/  "B0"; "A0"}; 
{\ar@/^0.4pc/ "B0"; "C0"}; 
{\ar@/^0.4pc/ "C0"; "B0"}; 
{\ar@/^0.4pc/ "C1"; "B0"}; 
{\ar@/^0.4pc/  "C0"; "D0"}; 
{\ar@/^0.4pc/  "D0"; "C0"};
{\ar@/^0.4pc/  "C0"; "D1"}; 
{\ar@/^0.4pc/  "D1"; "C0"};
{\ar@/^0.4pc/  "C1"; "D1"}; 
{\ar@/^0.4pc/  "D1"; "C1"};
{\ar@/^0.4pc/  "D1"; "E1"};
{\ar@/^0.4pc/  "E1"; "D1"}; 
{\ar@/^0.4pc/  "D1"; "E0"};
{\ar@/^0.4pc/  "E0"; "D1"}; 
{\ar@/^0.4pc/  "D0"; "E0"};
{\ar@/^0.4pc/  "E0"; "D0"}; 
{\ar@/^0.4pc/ "E0"; "F0"}; 
{\ar@/^0.4pc/ "F0"; "E0"}; 
{\ar@/^0.4pc/ "E1"; "F0"}; 
{\ar@/^0.4pc/ "F0"; "G0"}; 
{\ar@(ul,ur) "A0"; "A0"}; 
{\ar@(ul,ur) "G0"; "G0"}; 
(-60,-5)*{\scriptstyle \frac{1}{r+1}};
(-40,-11)*{\scriptstyle \frac{1}{r+1}};
(-30,-4)*{\scriptstyle \frac{r}{r+1}};
(-30,4)*{\scriptstyle \frac{1}{r+1}};
(-4,-10)*{\scriptstyle \frac 2 3   \frac{r}{r+1}};
(-5,-20)*{\scriptstyle \frac{1}{r+1}};
(-12,2)*{\scriptstyle \frac 1 3 \frac{r}{r+1}};
(-5,20)*{\scriptstyle \frac{r}{r+1}};
(-5,10)*{\scriptstyle \frac 1 4  \frac{1}{r+1}};
(3,-3)*{\scriptstyle \frac 3 4  \frac{1}{r+1}};
(25,-20)*{\scriptstyle \frac 2 3   \frac{1}{r+1}};
(24,-10)*{\scriptstyle \frac{r}{r+1}};
(26,20)*{\scriptstyle \frac 1 4  \frac{r}{r+1}};
(33,2)*{\scriptstyle \frac 3 4   \frac{r}{r+1}};
(26,10)*{\scriptstyle \frac{1}{r+1}};
(18,-3)*{\scriptstyle \frac 1 3   \frac{1}{r+1}};
(60,-11)*{\scriptstyle \frac{1}{r+1}};
(48,-4)*{\scriptstyle \frac{r}{r+1}};
(59,12)*{\scriptstyle \frac{r}{r+1}};
(80,5)*{\scriptstyle \frac{r}{r+1}};
\endxy
$
\caption{The reduced state space of $\hat \Proc^M$ on the octahedron $G = O_6$.}
\label{fig:octahedron2}
\end{figure}

\subsection{Monotonicity in $r$}
A consequence of G-symmetry for complete and cycle graphs is the fact that the maximum fixation time corresponds to the neutral drift $r =1$. The two previous examples, line and octahedron, show that this symmetry can partially persist, but it is globally broken (except in the two cases just mentioned). This implies a shift in the fitness value at which the maximum fixation time is reached, showing that the symmetry in Figure~\ref{fig:moran_conditional} is only apparent.  

\begin{proof}[Proof of Theorem~\ref{thm:decreaseMoran}] For any graph $G$ of order $N$, there are only a finite number of states in which the G-symmetry on $\calk_N$ can be broken. Therefore, the fitness value at which the maximum fixation time is reached can only be shifted finitely many times. 
For each vertex $v \in V$, as a rational function has finitely many local maxima, there is a fitness value $r(v) > 0$ such that the expected fixation time
$T^M_{\{v\}}(r)$ decrease for $r \geq r(v)$.
\end{proof}

The same argument applies for the mean fixation time $T_1^M$ (as convex combination of functions $T^M_{\{v\}}$) and the mean absorption time $\tau_1^M$ (as convex combination of the mean fixation and extinction times according to Supplementary Material S1). Obviously our argument does not allow us to say whether the maximum shifts to the right or left, nor the amplitude of the shift. For example, in the line $L_3$, the maximum is reached at $r = 1.4253$, whereas $r = 0.9997$ in the octahedron. In fact, for the graphs of order $6$, the fitness value where the maximum is reached ranges between $r=0.9991$ and $r =1.2559$. 

\subsection{Graphical model} 
To prove Theorem~\ref{thm:decreasebeta} about the mononicity of expected fixation time depending on the proliferation parameter, we now consider the proliferation process $\Proc = \Proc^\beta$ with $\beta = B,b$. Let us recall that Moran and proliferation processes can be obtained from a graphical representation  \cite{AlcaldeGuerberoffLozano2022}, which is inspired on Harris'  representation for contact processes \cite{Harris} (see also \cite{Griffeath}). This realisation is done on the set $G \times [0,+\infty)$ and we start by taking two independent Poisson distributions 
$$
\Pi_\varepsilon(v) = \{ \tau^\varepsilon_i(v) : i \in \N^\ast\}
$$
on $[0,+\infty)$, with rate $r$ for $\varepsilon = 1$ and rate $1$ for $\varepsilon = 2$. Each distribution determines a set of coloured marks
$$
\calm_\varepsilon = \{ (v,\tau) \in V \times [0,+\infty) : \tau \in \Pi_\varepsilon(v)\}, 
$$
red for $\varepsilon = 1$ and blue for $\varepsilon = 2$. Next we add two random sets of red and blue arrows. For any $(v,\tau) \in V \times [0,+\infty)$, we denote $A(v,\tau) \subset E \times \{v\}$ the set of edges with origin $v$. The set $\cala_1$ of red arrows is defined as follows: 
\begin{list}{\itemii}{\leftmargin=14pt}

\item[1)] For the Moran process, an element of $A(v,\tau)$ is chosen at random if $(v,\tau) \in \calm_1$. 

\item[2)] For Bernoulli proliferation, we have $\cala_1(v,\tau) = A(v,\tau)$ with probability $p$ and $\cala_1(v,\tau) = \emptyset$ with probability $1-p$.

\item[3)] For binomial proliferation, each element of $\cala_1(v,\tau)$ is chosen independently at random in $A(v,\tau)$ with probability $p$. 
\end{list}

\noindent
To define $\cala_2$, similarly to that happens in the Moran process, a unique element of $A(v,\tau)$ is chosen at random if $(v,\tau) \in \calm_2$. 

From this graphical model we obtain a continuous-time Markov chain $\{X_t\}_{t \in [0,+\infty)}$ which will realise the discrete-time Markov chain $\{S_n\}_{n \in \N}$ provided by $\Proc$. 
From now on, we will not distinguish between the expected absorption and fixation times of both types of processes, since they differ only in a factor that does not affect the inequalities that interest us. In fact, it will be convenient to consider only continuous-time processes, even if the results were stated for discrete-time processes. If $\Proc = \Proc^\beta$, we will write $\{X_t^\beta\}$ if we need to specify it.
Ignoring blue marks and arrows, we obtain another continuous-time Markov chain, which we will denote by $\{\bar X_t\}$ and call \emph{red process}. Up to scale, this process is uniquely determined by the underlying graph.
 
Processes with the same fitness, but associated to different proliferation parameters $0 < p' < p$ can be coupled taking the same random set of red and blue marks, but randomly deleting red arrows (see Supplementary Material S3 for details). Obviously, if the overall processes $\{X_t\}$ and $\{X'_t\}$ are coupled, then the red processes $\{\bar X_t\}$ and $\{\bar X'_t\}$ are also coupled.

For the red process, under Moran updating (1) and conditioning to an initial state $\bar X_0 = S_0 = \{v\}$, we define the first sojourn time as 
$$
\bar t_1 = \inf \{  t \in [0,+\infty) : \bar X_t \neq \bar X_0 \}
$$
Recursively, if $\bar t_n < +\infty$, we define 
$$
\bar t_{n+1} = \inf \{  t \in [\bar t_n,+\infty) : \bar X_t \neq \bar X_{\bar t_n} \}.
$$
All sojourn times
$$
\bar s_n = \bar t_n - \bar t_{n-1}
$$
follow exponential distributions, but the parameters depend on the states $S_{n-1}$. Given the random variable 
$$
\bar \n = \sup \{ n \geq 1 : \bar  t_n < +\infty \}, 
$$
the expected fixation time $T_{S_0}$ is the expectation of the random variable 
\begin{equation} \label{eq:absorptionred}
\bar  t_{\bar  \n} = \sum_{n=1}^\infty  \bar t_n \mathds{1}_{\{\bar \n=n\}}, 
\end{equation}
where $\mathds{1}_{\{\bar \n=n\}}$ is the characteristic function of the event $\{\bar  \n=n\}$. 
We have already described the 
mean fixation time $\bar T_1$ of the (neutral) red process for complete, cycle and star graphs, as it coincides with the limit $\lim_{r \to +\infty} \, T_1^M$ analytically computed in Remarks~\ref{rem:redcomplete},~\ref{rem:redcycle},~and~\ref{rem:redstar}. As observed in the Supplementary Material S3, if the sojourn times $\bar s_n$ are identically distributed, then Wald's first lemma implies that: 
\begin{equation} \label{eq:fixationred}
\bar T_{S_0} = \E[\bar t_{\bar \n} ] = \E[\bar t_1] \, \E[\bar \n].
\end{equation} 
More generally, returning to the overall process $\{X_t\}$, we can also define the first sojourn time 
$$
t_1 = \inf \{  t \in [0,+\infty) :  X_t \neq X_0 \},
$$
and recursively  
$$
t_{n+1} = \inf \{  t \in [t_n,+\infty) : X_t \neq X_{ t_n} \}
$$
if $t_n < +\infty$. The conditional distribution of $t_1$ given by $X_0 = S_0 = \{v\}$ is exponential with parameter $\lambda(S_0)$. Similarly 
$$
s_n = t_n - t_{n-1}
$$ follows an exponential distribution of parameter $\lambda(S_{n-1})$ when we conditioned to the state $S_ {n-1} = X_{t_{n-1}}$ in $\cals$. Thus, if $\{S_n\}$ denote the embedded jump chain associated to $\{X_t\}$ and if we conditioned to $X_0 = S_0, \dots,  X_{t_{n-1}} = S_ {n-1}$, the variable
$$
t_n =  s_1 + \dots +  s_n
$$
has a hypoexponential distribution with parameters $\lambda(S_0), \dots , \lambda(S_{n-1})$. This description should be compared to that of \cite[Section 4.8]{KT}). Defining 
$$
\n = \sup \{ n \geq 1 : t_n < +\infty \},
$$
then \eqref{eq:absorptionred} becomes: 
$$ 
t_\n = \sum_{n=1}^\infty  t_n \mathds{1}_{\{\n=n\}} 
$$ 
and \eqref{eq:fixationred} can be reformulated as follows (see again Supplementary Material S3): for any state $S_0 = \{v\}$, if the sojourn times $s_n$ are identically distributed, then the expected fixation time is given by 
\begin{equation} \label{eq:fixationoverall}
T_{S_0} = \E[t_\n | \F_{S_0}] = \E[t_1 | \F_{S_0}]\E[\n | \F_{S_0}],
\end{equation}
where $\F_{S_0}$ is the fixation event. Nevertheless, this approach to obtain monotonicity in $r$ and $p$ depends strongly on the hypothesis about the identical distribution of sojourn times.  
Furthermore, even under this constraint, it is not possible to prove monotonicity in $r$, not even for the expected value of the first sojourn time $\E[\bar t_1]$ or $\E[t_1 | \F_{S_0}]$. 

\subsection{Monotonicity in $p$}

Let $\Proc^\beta$ be the proliferation process on $G$ where $\beta = B,b$. As before under Moran updating, we can now define random variables $t_1^\beta$, $t_n^\beta$ and $t_\n^\beta$ associated to the continuous-time Markov chain $\{X_t^\beta\}$ under $\beta$ proliferation updating. The random variables associated to the corresponding red process will be denoted  $\bar t_1^\beta$, $\bar  t_n^\beta$ and $\bar t_{\bar \n}^\beta$. As before, we choose a vertex $v$ and take $S_0 = \{v\}$ as initial state. 

Given any fitness value $r > 0$, if the proliferation parameters $0 < p' \leq p$, then 
\begin{equation} \label{eq:redtimen}
\bar t_n^\beta(p') \geq \bar t_n^\beta(p)
\end{equation}
for any pair of coupled trajectories. 
Therefore
\begin{equation} \label{eq:expectedredtimen}
\E[\bar t_n^\beta(p')] \geq \E[\bar t_n^\beta(p)], 
\end{equation}
and then 
\begin{equation} \label{eq:redtime}
\E[\bar t_{\bar \n'}^\beta(p')] \geq \E[\bar t_{\bar \n}^\beta(p)]. 
\end{equation}

In the overall process, the blue marks and arrows modify the value of probabilities and times, but they do not modify the inequality after conditioning. 

\begin{proof}[Proof of Theorem~\ref{thm:decreasebeta}] Indeed, replacing the red process by the overall process and translating the conditioning by using fixation probabilities (as we describe in Supplementary Material S2), we deduce: 
\begin{align} \label{eq:monotonep}
\begin{split}
T_{S_0}^\beta(r,p')  & =  \frac{\E[t^\beta_{\n'}(p')]}{\Phi_{S_0}^\beta(r,p')}  \\ 
& \geq  \frac{\E[t^\beta_{\n'}(p')]}{\Phi_{S_0}^\beta(r,p)} \\
& \geq  \frac{\E[t^\beta_{\n}]}{\Phi_{S_0}^\beta(r,p)} = T_{S_0}^\beta(r,p), 
\end{split}
\end{align}
proving the result. 
\end{proof}


\section{Critical proliferation}
\label{Sphase}

In this section, we shall prove Theorem~\ref{thm:ptbeta}, that is, the existence of a second critical value of the proliferation parameter $p^\beta_t$ such that $\beta$ proliferation is advantageous in term of fixation time if $p \geq p^\beta_t$ and disadvantatgeous if $p \leq p^\beta_t$. The proof will be immediately derived from three lemmas. The first lemma is a corollary of Theorem~\ref{thm:decreasebeta}:  

\begin{lemma} \label{lem:1}
For any graph $G$ and any fitness value $r > 0$, the mean fixation rate $\Rho^\beta(r,p)$ is monotonically increasing in $p$.
\end{lemma}

\begin{proof} The mean fixation rate 
$$
\Rho^\beta(r,p) = \frac{T^M_1(r)}{T^\beta_1(r,p)}
$$
is obviously monotonically increasing in $p$ since the mean fixation time $T^\beta_1(r,p)$ is monotonically decreasing in $p$ according to Theorem~\ref{thm:decreasebeta}.
\end{proof}

From the previous remark, we know that both limits
$$
\lim_{p \to 0} T^\beta_1 (r,p) \quad \text{and} \quad \lim_{p \to 0}\Rho^\beta(r,p)
$$ 
always exist.

\begin{lemma} \label{lem:2}
For any graph $G$ and any fitness value $r > 0$, 
$\lim_{p \to 0}\Rho^\beta(r,p) < 1$.
\end{lemma}

\begin{proof} For each vertex $v \in V$, there exists $p > 0$ small enough to assure that the expected value of the sojourn time in the initial state $S_0 =\{v\}$ (which depends on $p$ when $\beta = B$, but also on $d(v)$ when $\beta = b$) is greater than the mean fixation time $T^M_1(r)$. 
As $V$ is finite, we deduce that all expected fixation times $T^\beta_{\{v\}} (r,p) > T^M_1(r)$ for $p$ small enough. It follows   
$$\lim_{p \to 0} T^\beta_1 (r,p) > T^M_1(r),$$
and then 
$$
\lim_{p \to 0}\Rho^\beta(r,p) < 1, 
$$
as we wanted. 
\end{proof}

%

\begin{lemma} \label{lem:3}
If $G \neq K_2$, then $\Rho^\beta(r,1) \geq 1$
\end{lemma}

\begin{proof} First, observe that the Moran process $\Proc^M$ and the proliferation process $\Proc^\beta$ can be coupled if $p = 1$. Furthermore, if $G \neq K_2$, the number of red arrows must be $\geq 2$ for some vertex $v \in V$.
This implies 
$T^\beta_{\{v\}} (r,1) \leq T^M_{\{v\}}(r)$
and 
then
$$
\Rho^\beta(r,1) = \frac{T^M_1(r)}{T^\beta_1(r,1)} \geq 1, 
$$
as we asserted. 
\end{proof}

\begin{proof}[Proof of Theorem~\ref{thm:ptbeta}] For any graph $G \neq K_2$ and any fitness value $r > 0$, the fixation time rate 
$\Rho^\beta(r,p) < 1$ for 
$p$ small enough from Lemma~\ref{lem:2} and $\Rho^\beta(r,1) \geq  1$ from Lemma~\ref{lem:3}. From Lemma~\ref{lem:1}, there exists a unique proliferation value $p_t = p^\beta_t(r) \in (0,1]$ such that 
$\Rho^\beta(r,p_t) = 1$. For $p < p_t$, the rate $\Rho^\beta(r,p_t) < 1$ and the proliferation is disadvantageous, and for $p > p_t$ the rate $\Rho^\beta(r,p_t) > 1$ and the proliferation is advantageous.
\end{proof}

\section{Conclusion}

\emph{Proliferation processes} were introduced in \cite{AlcaldeGuerberoffLozano2022} as a realistic method for increasing the fixation probability of mutant individuals in populations structured on graphs. However, somewhat surprisingly, we found that there always is a value of the proliferation parameter depending on the graph structure in which the fixation probability was the same. For lower values, proliferation is disadvantageous and only for higher values is advantageous. However, it could be expected that the occupation of neighbouring nodes following Bernoulli or binomial draws would always be advantageous in terms of (mean) fixation time. 

This idea was initially supported by the analytical computation of mean absorp\-tion and fixation times for Bernoulli proliferation on complete, cycle and star graphs in Section~\ref{SBernoulli} and binomial proliferation on complete graphs in Section~\ref{Sbinomial}. We have applied methods similar to those used in \cite{Antal2006}, \cite{BHR} and \cite{H} to solve the equations formulated from the transition matrix of the Moran process. Exact computations of mean fixation times for all graphs of order 6 (included in Supplementary Material) raised two questions about this process: first the apparent symmetry between the mean fixation time for mutants with fitness $r$ and $1/r$ and second the monotonicity of the mean fixation time when $r > 1$. As observed by Hathcock and Strogatz in \cite{HathcockStrogatz}, this symmetry can be derived from Maruyama-Kimura symmetry, introduced in  \cite{MaruyamaTimes1974} and \cite{MaruyamaKimura} and  rediscovered in \cite{Antal2006}, \cite{Traulsen2015} and \cite{Taylor2006}. We derive this symmetry from some symmetry of the process, which we has called \emph{G-symmetry}, only observed in cliques and cycles. As a consequence, the maximum fixation time for these graphs is reached when $r=1$, but the breaking of the G-symmetry implies a shift to the right or the left. Our method does not allow us to compute the displacement, nor predict its direction. We only know that the mean fixation time is monotonically decreasing from this value. It would be interesting to find another method that explains this shift and, more precisely, why the values found by numerical methods for regular graphs of order 6 are less than 1.

Contrary to the initial idea, supported by analytical calculations for cliques, cycles and stars, we have proven that there is a second critical value $p_t$ such that proliferation is no longer advantageous for $p \leq p_t$. As in \cite{AlcaldeGuerberoffLozano2022}, we have resorted to a graphical model based on the Harris' graphical method \cite{Harris} and reduced the analysis to a "no-backward" process that we have called \emph{red process}. This method has allowed us to demonstrate that the mean fixation time decreases monotonically when the proliferation parameter increases, providing the second critical value, but does not explain 
the entire dynamics of the process. An open question is the existence of graph structures where proliferation is advantageous in probability and disadvantageous in time in the intermediate regime. Another issue to be studied in the future is the relation between the fixation time for the (neutral) red process and the cover time for the SRW on the underlying graph.

Our main contribution in this work has been to clarify the role of symmetries of the Moran process for computing fixation time in a population structured on a graph. Unlike what happens with the fixation probability, the Maruyama-Kimura symmetry and the symmetry for mutants with fitness $r$ and $1/r$ depends on the graph structure induced on each states,  justifying 
some asymmetry in all graphs except cliques and cycles. The application of Harris' graphical method to the study of Moran and proliferation processes is promising, but we have only used it to study the dependence of the mean fixation time on the proliferation parameter. 
It is possible that the combination of these methods with others, whether classical (such as those used by Durrett, Grimmett and their co-authors \cite{Durrettcontact1, Durrettcontact2, Durrettcontact3,Grimmettcontact} in the study of contact processes) or more recent (such as the martingale method proposed by Monk and van Schaik in \cite{Monk2020} and \cite{Monk2021}), will allow progress in the study of the absorption and fixation times of the Moran process.

\bibliographystyle{unsrt}

\end{document}